\documentclass[aps,pra,twocolumn,floatfix,groupedaddress,superscriptaddress,footinbib,notitlepage,showkeys,showpacs,nofootinbib]{revtex4-1}
\usepackage{graphicx,graphics,times,bm,bbm,bbold,amssymb,amsmath,amsfonts,dsfont,color,xcolor,dcolumn}
\usepackage[bookmarksnumbered, colorlinks,plainpages]{hyperref}
\usepackage{diagbox}
\usepackage{multirow}
\usepackage{mathtools}
\hypersetup{colorlinks,linkcolor=blue,citecolor=blue,urlcolor=blue}
\usepackage{amsthm}

\newcommand{\ket}[1]{|#1\rangle}
\newcommand{\bra}[1]{\langle #1|}

\newcommand{\Id}{{\mathbb I}}
\newcommand{\Tr}{{\textrm {Tr}}}
\newcommand{\T}{{\mathrm {t}}}
\newcommand{\e}{{\mathrm {e}}}
\newcommand{\diag}{{\mathrm {diag}}}
\newcommand{\etal}{\textit {et al. }}

\theoremstyle{definition}

\newtheorem{theorem}{Theorem}
\newtheorem{lemma}[theorem]{Lemma}
\newtheorem{proposition}[theorem]{Proposition}

\newtheorem{remark}[theorem]{Remark}

\def \l{\mathcal{L}}
\def \d{\mathcal{D}}
\begin{document}

\title{
Optimal quantum speed for mixed states
}
\

\author{Ashraf Naderzadeh-ostad}
\email{a\_naderzadeh@outlook.com}
\affiliation{Department of Physics, Ferdowsi University of Mashhad, Mashhad, Iran}
\author{Seyed Javad Akhtarshenas}
\email{akhtarshenas@um.ac.ir}
\thanks{Corresponding author}
\affiliation{Department of Physics, Ferdowsi University of Mashhad, Mashhad, Iran}

\begin{abstract}
The question of  how fast a quantum state can evolve is considered. Using the definition of squared speed based on the Euclidean distance given in [Phys. Rev. Research, {\bf 2}, 033127 (2019)], we present a systematic framework to obtain the optimal speed of a  $d$-dimensional system evolved unitarily under  a  time-independent Hamiltonian.
Among the set of mixed quantum states having the same purity, the optimal state is obtained in terms of its purity parameter.  We show that for  an arbitrary $d$, the optimal state is represented by a  $X$-state with an additional property of being  symmetric with respect to the secondary diagonal.
For  sufficiently low purities for which the purity  exceeds the purity of maximally mixed state $\Id/d$ by at most $2/d^2$,  the only nonzero off-diagonal entry of the optimal state is  $\varrho_{1d}$, corresponding to the transition amplitude between two energy eigenstates with minimum and maximum eigenvalues, respectively.  For larger  purities, however,  whether or not the other secondary diameter entries $\varrho_{i,d-i+1}$  take nonzero values depends on their  relative   energy gaps $|E_{d-i+1}-E_{i}|$.
The effects  of coherence and entanglement, with respect to the energy basis,  are  also examined and found that for optimal states both resources are monotonic functions of purity, so  they can cause speed up quantum evolution leading to a smaller quantum speed limit.
Our results show that although the coherence of the states is responsible for the speed of evolution,   only the coherence caused by some off-diagonal entries located on the secondary diagonal play a  role in the fastest states.
\end{abstract}
\keywords{quantum speed limit; optimal quantum speed; mixed states; Euclidean metric; Wigner-Yanase skew information}

\maketitle

\section{Introduction}\label{Sec-Introduction}
The maximum  speed of evolution of a quantum system  is a fundamental and important concept that has many applications in many areas of  physics such as quantum communication \cite{BekensteinPRL1981}, quantum computation \cite{LloydN2000}, quantum metrology \cite{GiovannettiNP2011}, and optimal control \cite{CanevaPRL2009}.  The quantum speed limit (QSL)  then provides a lower bound on the  time required   to evolve an initial state $\varrho_{0}$ to a target state $\varrho_{t}$ under the Hamiltonian $H$
\begin{equation}
\tau\ge \tau_{\textrm{QSL}}.
\end{equation}
For isolated systems, the QSL between two  orthogonal states is defined as the maximum of the Mandelstam-Tamm (MT) bound  \cite{MandelstamJP1945}  and the Margolus-Levitin (ML) bound \cite{MargolusPD1998}, i.e.,
\begin{equation}
\tau_{\textrm{QSL}}=\max\{\pi\hbar/(2\Delta E), \pi \hbar/(2\langle E \rangle) \}.
\end{equation}
Above, the MT bound depends on the variance of the energy  $\Delta E$ of the initial state and the ML bound  depends on the mean energy $\langle E \rangle$ with respect to the ground state, i.e.,  under the assumption that the energy of the ground state is  taken to be zero.
Levitin and Toffoli  proved that the combined bound of MT and ML is tight and can only be attained by two-level state  $\ket{\psi}=\frac{1}{\sqrt{2}}(\ket{E_{1}}+\e^{-i\theta} \ket{E_{2}})$ where $H \ket{E_{i}}=E_{i}\ket{E_{i}}$ and  $\theta$ is an arbitrary phase \cite{LevitinPRL2009}. They have also shown  that no mixed state can have a larger speed.

The QSL  was originally derived for the unitary evolution of pure and orthogonal states and is then generalized to the  arbitrary angles \cite{GiovannettiPRA2003}, mixed states \cite{GiovannettiPRA2003,WuPRA2018,CampaioliPRL2018}, and nonunitary evolutions \cite{TaddeiPRL2013,delCampoPRL2013,DeffnerPRL2013}.
There are also some researches   on the relationship between the QSL and the quantum phenomena such as the role of entanglement  in dynamical evolution \cite{GiovannettiEU2003,BatlePRA2005,ZanderJPA2007,FrowisPRA2012,BorrasPRA2006,KupfermanPRA2008}, non-Markovian effect of the environment on the speed of quantum evolution \cite{DeffnerPRL2013,CianciarusoPRA2017},
and dependence of the QSL on the initial state \cite{WuJPA2015,LiuPRA2015}.  Very recently, the QSL  has also been used to derive lower bounds on the minimal time required to vary (generation  and degradation) a quantum resource  \cite{CampaioliNJP2022}, the  informational measures  \cite{MohanNJP2022},  and  the correlations \cite{PandeyPRA2023} of quantum systems under some quantum dynamical process. It was also used to  interpret  the  geometric measure of multipartite entanglement for pure states as the minimal time necessary to unitarily evolve a given quantum state to a separable one \cite{RudnickiPRA2021}.

Among the three elements of $ \tau_{\textrm{QSL}}$, namely $\varrho_{0}$, $\varrho_{t}$ and $H$, in most QSL discussions it is assumed that the Hamiltonian of the system is given and  by fixing the distance between the initial and final states, the evolution time is bounded for any arbitrary state according to the chosen metric \cite{FreyQINP2016}. The MT and ML bounds are of this type.
The evolution time, however,  can be framed and studied in various other ways. In quantum control, for instance, the task is to search for the  optimal Hamiltonian  to drive  the physical system   between given initial and  final quantum states in the shortest time.  Finding the   time-optimal evolution is analogous to that for the brachistochrone in classical mechanics and determines the optimal Hamiltonian and the type of control as the  physical   resources needed  to carry out the process \cite{CarliniPRL2006,CarliniJPA2008,WangPRL2015,GengPRL2016,LamPRX2021}.

For the construction of QSL bounds,  choosing a measure of distinguishability is, in general, necessary in order to be able to define the quantum speed as the derivative of some distance \cite{DeffnerJPA2017}. Various metrics have been used to obtain  better bounds for QSL, for instance, del Campo \etal  used the relative purity \cite{delCampoPRL2013}, Deffner and Lutz worked with Bures angle \cite{DeffnerPRL2013}, and Campaioli \etal  defined a new distance based on the angle between generalized Bloch vectors  \cite{CampaioliPRL2018}. Pires \etal  have  shown that there is an infinite family of contractive metrics that can be used to obtain QSL bounds \cite{PiresPRX2016}. Such  bounds depend on the chosen metric and, in general, it is not easy to obtain a tight and attainable bound by the appropriate choice of the metrics. In more general
cases, QSLs are loose for mixed states of unitary evolutions or open dynamics \cite{CampaioliQ2019,PiresPRX2016}.   In particular,  for pure states both  MT and  ML QSL bounds are attainable, however,  for a general mixed state they can be rather loose.  For instance, a state that  is diagonal  in the energy eigenbasis does not evolve further, so  the speed of evolution is zero or, equivalently, $\tau_{\textrm{QSL}}$ is infinite.  However, in this case,  both the MT and ML lower bounds on $\tau_{\textrm{QSL}}$  can be arbitrarily small. This means,  for incoherent states in the energy eigenbasis,   $\Delta E$ and $\langle E \rangle  $ do not contain any  information about the speed of evolution \cite{MarvianPRA2016}.  In addition, any function that is used to  quantify the speed of evolution is expected to be nonincreasing under mixing, however, the standard MT bound does not have this property, because the uncertainty $\Delta E$ in general increases under mixing.
Motivated by these,  Marvian \etal have found  QSL bounds  in terms of the coherence of states relative to the energy eigenbasis \cite{MarvianPRA2016}.
Also, in  \cite{PiresPRX2016,MondalPLA2016}  the authors   obtained bounds for $\tau_{\textrm{QSL}}$ in terms of skew information, as a  measure of coherence.
Another problem related to QSL bounds is the explicit dependence of many   bounds on the time parameter \cite{ShaoPRR2020}.   For a noncontrolled fixed Hamiltonian  and by starting from a given   initial state  $\varrho_{0}$,   the dynamical trajectory  as well as the set of states satisfying the target angle is fixed, so that the QSL  should not  depend on the time either. Motivated by this, the authors of Ref. \cite{ShaoPRR2020} have defined the QSL time  as the minimum evolution time to fulfill the target angle $\Theta$ for any initial state $\varrho_{0}$  that can fulfill the target angle $\Theta$.

The above discussion  implies that the QSL requires more considerations, in particular for mixed states,  so  new approaches in the field can be effective and even necessary.  Our main goal in this study is to find a tight and maximum speed for mixed states of the isolated  systems. We survey directly on the evolution speed instead of the evolution time, therefore, we do not face the problem of the explicit dependence of the QSL on the time parameter.
The notion of speed depends on the choice of metric on the space of quantum states, each of which has its own figure of merit depending on the desired purpose and application. In this work, we adapt  the notion of speed based on the Euclidean metric defined naturally  on  the Hilbert-Schmidt  space in  the Bloch vector representation of the  quantum states \cite{BrodyPRR2019}. This choice of metric provides a simple form for speed in terms of  $\varrho$ instead of $\sqrt{\varrho}$ which appears in the squared speed based on the Wigner-Yanase (WY) skew information \cite{MondalPLA2016, PiresPRX2016}, so making   our calculations computable.  The defined  speed vanishes for all incoherent states, and even more, it is convex under mixing  thus speed does not increase under mixing.

We restrict  ourselves to the unitary evolution generated by the time-independent Hamiltonian. In this case,  the speed remains constant in time, regardless of what initial state the system starts in and how far the distance between $\varrho_{0}$ and $\varrho_{t}$ is. The speed therefore depends  only on the initial state. Our aim is  to express  this dependence in terms of purity; for instance, zero speed for maximally mixed state and maximum speed  for the pure state $\ket{\Psi}=\frac{1}{\sqrt{2}}(\ket{E_{\min}}+\e^{-i\theta}\ket{E_{\max}})$.  Motivated by this, we pursue then the question that among the set of all quantum states with a given definite   purity,  which  initial state represents the optimal one. By parameterizing  an arbitrary $d$-dimensional mixed state in terms of its  purity, we identify the state with  maximum speed for each given purity.   Such obtained optimal state provides a tight bound for the QSL induced by  the Euclidean metric. We apply our  optimal states to  the squared speed defined by the WY skew information \cite{MondalPLA2016, PiresPRX2016}. The numerical simulations show that the proposed  optimal state provides almost a faster squared speed, except for some region of low purity for which the speed exceeds the speed limit proposed by our optimal state. In such cases, however, our optimal state  can be exploited in order to derive  lower bounds for the maximum speed.

Our other important result is that although the off-diagonal entries of a quantum state  play a key role both in the quantum coherence of the state  and in the  speed of quantum evolution, their influences do not match with each other  in general. In other words,  despite the coherence of the initial state can speed up the quantum evolution process,  for the optimal states only the coherence arisen from the secondary diameter of the state can have  a nonzero contribution in the optimal speed, all other off-diagonal entries  have vanishing  contributions. Furthermore, for all density matrices with arbitrary  purity, the  entry $\varrho_{1d}$ always has a nonzero contribution on the optimal state. However, whether or not the other secondary diameter entries $\varrho_{i,d-i+1}$  take nonzero values depends on their  relative   energy gaps $|E_{d-i+1}-E_{i}|$.

For those dimensions $d$ characterizing  a bipartite $d_1\times d_2$ system  with $d=d_1d_2$, we investigate the behavior of the optimal speed with respect to the entanglement of the optimal state. In doing so, we suppose the energy eigenstates provide a product basis, and we find that except for the sufficiently low purities, the optimal state has nonzero entanglement. Our results show in particular that for the case of an $N$-qubit system for which $d=2^N$, the one-qubit reduced state of the optimal state, obtained by tracing out $N-1$ qubits,  is a maximally mixed state.  When  $d=4$, for instance,  the optimal states therefore form a special class of Bell-diagonal states, up to a local unitary transformation.       Furthermore,  for optimal states both resources are monotonic functions of purity, so  they can cause speed up quantum evolution leading to a smaller quantum speed limit.
\\
The paper is organized as follows. In Section \ref{Sec-Speed}, we first briefly review the squared  speed of evolution of a $d$-dimensional system  with respect to the Euclidean metric  defined in \cite{BrodyPRR2019}. Some related properties and preliminaries are given also in this section. Section \ref{Sec-Optimal} is devoted on the optimal speed of a mixed state with an arbitrary purity. Our main results, namely Proposition  \ref{Prop-Xstate} and Theorems  \ref{Theorem-Opt-k0}, \ref{Theorem-Opt-Main-1} and \ref{Theorem-Opt-Main-2} are presented in this section, however, some of their proofs can be found   in Appendices \ref{Appendix-Prop-Xstate} and \ref{Appendix-Theorem-Opt-Main-1,2}.     A simulation evidence of our results is provided in Sec. \ref{Sec-Simulation}.  In this section we also apply our optimal states to the squared speed defined by WY skew information. The role of entanglement and coherence of the optimal states are discussed in Sec. \ref{Sec-Discussion}.   The paper is concluded with a brief review in Sec. \ref{Sec-Conclusion}.

\section{Speed of quantum evolution}\label{Sec-Speed}
An arbitrary  quantum state  $\varrho$ acting on the  $d$-dimensional Hilbert space  $\mathcal{H}$ can be written in the Bloch vector representation as
\begin{equation}\label{eq10.1}
\varrho=\frac{1}{\sqrt{d}}\sigma_0 +\sum_{j=1}^{d^2-1} r_j \sigma_j,
\end{equation}
where $\sigma_0=\Id/\sqrt{d}$, and  $\{\sigma_j\}_{j=1}^{d^2-1}$ are Hermitian traceless orthonormal basis,
with the Hilbert-Schmidt inner product  $\langle \sigma , \tau\rangle_{\textrm{HS}} =\Tr[\sigma^{\dagger} \tau]$. In this representation $r_j=\Tr[\varrho\sigma_j]$, and  $\boldsymbol{r}=(r_1,\cdots,r_{d^2-1})^{\T}\in\mathbb{R}^{d^2-1}$  is   the Bloch vector corresponding to the state $\varrho$.

Starting from  the initial state $\varrho_0$,  consider a one-parameter family of states $\varrho_t$  generated by an  open system dynamics
as
\begin{eqnarray}\label{eq10.2}
\partial_t \varrho =\l \varrho =-i [H,\varrho] +\d \varrho,
\end{eqnarray}
where the dissipator term
\begin{equation}\label{eq10.5}
\d \varrho =\sum_k \Big[ L_k \varrho L_k^{\dagger} -\frac12 (L_k^{\dagger} L_k \varrho +\varrho L_k ^{\dagger} L_k)\Big],
\end{equation}
is responsible for the open dynamics of the system.
Based on the above dynamics, the authors of Ref. \cite{BrodyPRR2019} obtained the  squared speed of evolution with respect to the Euclidean metric as
\begin{equation}\label{eq10.3}
v^2(\varrho)=\sum_{j=1}^{d^2-1} \dot{r}_j^2 =\Tr[(\l \varrho)^2],
\end{equation}
which   can be written as \cite{BrodyPRR2019}
\begin{eqnarray}\label{SqSpeed}
v^2(\varrho)=2 \Tr[H[H,\varrho]\varrho]+\Tr[(\d \varrho)^2]
- 2i  \Tr[\varrho[\d \varrho, H]].
\end{eqnarray}
Clearly, the first and the second terms arise  purely from unitary and dissipative dynamics of the evolution, respectively.
The third  term, on the other hand,  comes from incompatibility between the Hamiltonian $H$ and the set of Lindblad operators $\{L_k\}$ \cite{BrodyPRR2019}.

\subsection{Speed of unitary evolution}\label{se11}
In this work,  we consider only unitary evolutions. We also assume Hamiltonian is \textit{time-independent}. With these assumptions, the squared speed for a closed system becomes time-independent and depends only on the \textit{initial state}. So, if we find the optimal  state initially, it remains optimal at later times as well.

Let  $\{\ket{E_n}\}_{n=1}^{d}$ denotes  the orthonormal eigenbasis of the Hamiltonian $H$ corresponding to the eigenenergies $\{E_n\}_{n=1}^{d}$, i.e. $H=\sum_{n=1}^{d}E_n\ket{E_n}\bra{E_n}$. In this basis,   $\varrho$ is represented by $\varrho=\sum_{i,j=1}^{d}\varrho_{ij} \ket{E_i}\bra{E_j}$ and  Eq. \eqref{SqSpeed} reads
\begin{eqnarray} \nonumber
v^2(\varrho)&=& 2 \left\{\Tr[H^2\varrho^2]-\Tr[H\varrho H\varrho]\right\} \\ \label{SqSpeed-Unitary-1}
&=& \sum_{i,j} |\varrho_{ij}|^2 \omega_{ij}^2.
\end{eqnarray}
Above,  in the natural units with $\hbar=1$, we have defined  $\omega_{ij}=(E_j-E_i)>0$ for $1\le i< j\le d$ as  the Bohr frequency corresponding to the transition  between lower  and upper levels $E_{i}$ and  $E_{j}$, respectively. Clearly, $\omega_{ij}$ obeys the sum rule $\omega_{ij}=\sum_{k=i}^{j-1}\omega_{k,k+1}$.
It follows from above relation that   \textit{only} the off-diagonal elements of $\varrho$, in the Hamiltonian bases, contribute to the speed of unitary evolution.
Before we proceed further, it may be relevant to point out some  properties of the squared speed.

\subsection{Properties of the squared speed}
\begin{lemma}
Let $\mathcal{H}_1$ and $\mathcal{H}_2$ be two complementary orthogonal subspaces of $\mathcal{H}$  on which  the Hamiltonian $H$ is partitioned into $H_1$ and $H_2$,  i.e.,  $\mathcal{H}=\mathcal{H}_1 \oplus \mathcal{H}_2$ and $H=H_1+ H_2$. Then for any pair of orthogonal states $\varrho_1$ and $\varrho_2$  acting  on $\mathcal{H}_1$ and $\mathcal{H}_2$, respectively,  we have
\begin{equation}\label{SqSpeed-DirectSum}
v^{2}(\lambda \varrho_1+(1-\lambda)\varrho_2)=\lambda^2 v^{2}(\varrho_1)+(1-\lambda)^2 v^2(\varrho_2).
\end{equation}
\end{lemma}
\begin{proof}
Putting $\varrho=\lambda\varrho_{1}+(1-\lambda)\varrho_{2}$ and $H=H_1+H_2$ in the first line of Eq. \eqref{SqSpeed-Unitary-1}, and using the fact that $\varrho_i H_j=0$ for $i\ne j$,  one can easily obtain the equation above.
\end{proof}
\begin{lemma}
The squared speed of unitary evolution is a convex function of  $\varrho$, i.e., for any   pair  of states $\varrho_1$ and $\varrho_2$ and an arbitrary $0\le \lambda\le 1$, we have
\begin{equation}\label{SqSpeed-Convexity-1}
v^{2}(\lambda \varrho_1+(1-\lambda)\varrho_2)\le\lambda v^{2}(\varrho_1)+(1-\lambda)v^2(\varrho_2).
\end{equation}
Moreover, the inequality   is saturated if and only if,  in the Hamiltonian basis, the  off-diagonal elements of $\varrho_1$ and $\varrho_2$  are likewise equal,  $(\varrho_1)_{ij}=(\varrho_2)_{ij}$ for $i\ne j$.
\end{lemma}
\begin{proof}
To prove the lemma, note that  the first line of Eq. \eqref{SqSpeed-Unitary-1}  can be written also as
\begin{equation}\label{SqSpeed-Unitary-3}
v^{2}(\varrho)=-\Tr\left[ H,\varrho\right] ^{2}=\| [H, \varrho ] \|_{\textrm{HS}}^2,
\end{equation}
where   $\|A\|_{\textrm{HS}}=\{\Tr[A^\dagger A]\}^{1/2}$ is  the Hilbert-Schmidt norm induced from the Hilbert-Schmidt inner product $\langle A, B\rangle_{\textrm{HS}}=\Tr[A^\dagger B]$. In view of this, the convexity of the squared speed follows from the convexity of the  Hilbert-Schmidt norm, i.e., for any pair of operators $A$ and $B$,  and  $0\le \lambda\le 1$,  we have
\begin{eqnarray}\label{SqSpeed-Convexity-3}
\lambda \| A \|_{\textrm{HS}}^2+(1-\lambda)\| B \|_{\textrm{HS}}^2 &-&\| \lambda A+(1-\lambda) B \|_{\textrm{HS}}^2 \\ \nonumber
&=&\lambda(1-\lambda) \|A-B \|_{\textrm{HS}}^2\ge 0,
\end{eqnarray}
Inserting   $A=[H, \varrho_1 ]$ and $B=[H, \varrho_2 ]$, we arrive at our  assertion   \eqref{SqSpeed-Convexity-1}.

From  $A-B=[H,\varrho_1-\varrho_2]$, it follows  that the inequality  \eqref{SqSpeed-Convexity-3} is saturated if and only if $[H,\varrho_1-\varrho_2]=0$, i.e.,  if and only if $\varrho_1-\varrho_2$ is diagonal in the Hamiltonian basis, implying that the  off-diagonal elements of $\varrho_1$ and $\varrho_2$  are likewise equal,  $(\varrho_1)_{ij}=(\varrho_2)_{ij}$ for $i\ne j$. Note that, as only off-diagonal elements of a density matrix play a role in the speed,  we have  therefore $v^{2}(\varrho_1)=v^{2}(\varrho_2)$, however, the converse is not true in general, i.e., it may happens that $v^{2}(\varrho_1)=v^{2}(\varrho_2)$ but the inequality  \eqref{SqSpeed-Convexity-3} is not saturated.
\end{proof}
As we mentioned previously in Sec. \ref{Sec-Introduction},  this feature of convexity is crucial for any function that is used to quantify the speed of evolution in order to be nonincreasing under mixing \cite{MarvianPRA2016}. In Sec. \ref{Sec-Optimal} we see that the optimal speed is a monotonic function of the purity of the state.
\begin{lemma}\label{Prop-SqSpeed-UBound}
For a general state $\varrho$, the squared speed is bounded from above by
\begin{eqnarray}\label{SqSpeed-Unitary-3}
v^2(\varrho)  \le   2\left(\Delta H\right)^2,
\end{eqnarray}
where  $(\Delta H)^2=\left<H^2\right>-\left<H\right>^2$ is the variance of $H$.
Moreover, the inequality is saturated if and only if the state is pure (see \cite{GessnerPRA2018} for a same relation for quantum Fisher information).
\end{lemma}
\begin{proof}
In the Hamiltonian basis, the variance is expressed as  $(\Delta H)^2=\sum_{i<j}\varrho_{ii}\varrho_{jj}\omega_{ij}^2$. Using this and  defining
\begin{equation}\label{Mu-ij}
\mu_{ij}=\varrho_{ii}\varrho_{jj}- \vert \varrho_{ij} \vert ^{2},
\end{equation}
for $1\le i< j\le d$,  Eq. \eqref{SqSpeed-Unitary-1} takes the following form
\begin{eqnarray}\label{SqSpeed-Unitary-4}
v^2(\varrho) = 2\left(\Delta H\right)^2-2\sum_{i<j}\mu_{ij}\omega_{ij}^2.
\end{eqnarray}
As $\varrho\ge 0$, then $\mu_{ij}\ge 0$ for all $i,j=1,\cdots,d$, leads therefore to  the inequality \eqref{SqSpeed-Unitary-3}.  The proof is complete if we recall  that a state is pure if and only if $\mu_{ij}=0$ for all $i,j=1,\cdots,d$.
\end{proof}

\section{Optimal-speed quantum states}\label{Sec-Optimal}
Under  the Hamiltonian $H$, how fast can a  quantum state  evolve in time? In particular, which  of the pure states is the fastest one or, even more, over the set of  density matrices with equal purity, which states have  the optimum speed? In what follows, we are going to address  these questions by developing a  framework to obtain the optimal state having definite  purity.
\subsection{Optimal pure quantum states}
According to the Lemma \ref{Prop-SqSpeed-UBound}, the squared speed of a pure state $\varrho=\ket{\psi}\bra{\psi}$ reduces  to
\begin{eqnarray}\label{SqSpeed-Unitary-Pure}
v^2(\varrho)=2 \left(\Delta H\right)^2.
\end{eqnarray}
In view of this, the following proposition introduces  the states for which the energy variance is maxima \cite{TextorIJTP1978,Sakmann2011PRA}, so the  squared speed is maximum.
\begin{proposition}\label{Prop-SqSpeed-Unitary-Pure}
Let  $\{\ket{E_1},\ket{E_2},\cdots,\ket{E_d}\}$ denotes  the orthonormal eigenbasis of the Hamiltonian $H$, corresponding  to the nondecreasing ordered  eigenvalues $E_{1}\le E_2\le \cdots \le E_{d}$.  Then the pure state $\ket{\Psi}=\sum_{n=1}^{d}\alpha_n\ket{E_n}$ makes the energy variance maximum if
$\ket{\Psi}=\frac{1}{\sqrt{2}}\left(\ket{E_{\min}}+\e^{-i\theta}\ket{E_{\max}}\right)$, where $\ket{E_{\min}}$ and $\ket{E_{\max}}$ are two states in the eigensubspaces of minimum and maximum energies $E_1$ and $E_d$, respectively. Accordingly, the maximum value of the variance of $H$ and its corresponding  maximum squared speed are  $\Delta H_{\max}^2=\frac{1}{4}\left( E_{1}-E_{d}\right) ^{2}=\frac{1}{4}\omega_{1d}^2$ and $v^2_{\max}(\Psi)=\frac{1}{2}\omega_{1d}^2$, respectively.   Obviously, when both $E_1$ and $E_d$ are nondegenerate,  the maximum variance  happens for $\ket{E_{\min}}=\ket{E_{1}}$ and $\ket{E_{\max}}=\ket{E_{d}}$, i.e. when $\vert \alpha_{1}\vert^{2}=\vert \alpha_{d}\vert^{2}=\frac{1}{2}$.
\end{proposition}

\subsection{Optimal mixed quantum states}
For a general quantum state $\varrho$, acting on the $d$-dimensional Hilbert space $\mathcal{H}$, one can define the purity as $\kappa=\Tr[\varrho^2]$, where ranges from $\kappa=1/d$ for the maximally mixed state $\varrho_0=\Id/d$ to $\kappa=1$ for an arbitrary pure state corresponding to the projection operator  $P_{\psi}=\ket{\psi}\bra{\psi}$. In what follows, we are looking for a class of mixed states $\varrho_{\kappa}$, parameterized by the associated  purity $\kappa\in[1/d,1]$, such that the squared speed is maximal. For $\kappa=1$, the result of the Proposition \ref{Prop-SqSpeed-Unitary-Pure} is  retrieved by a new approach.

In  energy eigenbasis, the purity of a general state $\varrho_{\kappa}=\sum_{ij}\varrho_{ij} \ket{E_i}\bra{E_j}$ can be  expressed as
\begin{eqnarray}\label{Purity-1}
\kappa=\sum_{i=1}^{d} \varrho_{ii}^{2}+\sum_{i\neq j}^{d} \vert \varrho_{ij} \vert ^{2}.
\end{eqnarray}
For the maximally mixed state for which the speed is zero, we have  $\kappa=1/d$. It follows therefore that in order to have nonzero speed, we have to  increase the purity of   state. Our goal is  to increase the speed by increasing  the purity in an optimum manner.
A systematic framework to obtain the optimal states of a  $d$-dimensional system is presented by  Theorems  \ref{Theorem-Opt-k0}, \ref{Theorem-Opt-Main-1} and \ref{Theorem-Opt-Main-2} below. Their structural form, however, is given by the following proposition.

\begin{proposition}\label{Prop-Xstate}
For an arbitrary dimension $d$, the optimum state is given by a persymmetric  $X$-state, i.e., a $X$-state which is symmetric with respect to the secondary diagonal so that  $\varrho_{ij}=\varrho_{d-j+1,d-i+1}$ for all $i,j$
\begin{equation}\label{Optimal-X-state}
\varrho_{\kappa}=\begin{pmatrix}
\varrho_{11} & 0   & \cdots & 0 &   \varrho_{1d} \\
 0 & \varrho_{22} &  \cdots &   \varrho_{2,d-1} & 0 \\
 \vdots & \vdots & \ddots & \vdots &   \vdots \\
0 & \varrho_{2,d-1}^\ast &  \cdots & \varrho_{22} & 0 \\
\varrho_{1d}^\ast & 0  & \cdots & 0  & \varrho_{11} \\
\end{pmatrix}.
\end{equation}
Here,  in the innermost we have   $\varrho_{\frac{d+1}{2},\frac{d+1}{2}}$ when $d$ is odd, and a two-dimensional persymmetric matrix
when $d$ is even.
\end{proposition}
\begin{proof}
The proof is provided in Appendix \ref{Appendix-Prop-Xstate}.
\end{proof}

The following theorems determine the nonzero entries of the optimal state \eqref{Optimal-X-state} for various values of purity.
\begin{theorem}\label{Theorem-Opt-k0}
For $\kappa\in\left[1/d,\kappa_0\right]$ where
\begin{equation}\label{kappa0}
\kappa_0=1/d+2/d^2,
\end{equation}
the optimal quantum state  is given by
\begin{eqnarray}\label{rhoOP-0}
\varrho_{\kappa\in[\frac{1}{d},\kappa_0]}=\frac{1}{d}\sum_{i=1}^{d} \ket{E_i}\bra{E_i}+\varrho_{1d}\ket{E_1}\bra{E_d}+\varrho_{1d}^\ast\ket{E_d}\bra{E_1},
\end{eqnarray}
where $\varrho _{1d}= \e^{i\theta_1}\sqrt{\frac{1}{2}\left( \kappa-1/d\right)}$.
Obviously,  the rank of the optimal  state \eqref{rhoOP-0} is full except at $\kappa_0$ for which the rank is diminished by 1, so that $\textrm{rank}\{\varrho_{\kappa_0}\}=d-1$. Moreover, the optimal squared speed is $v^2(\varrho_{\kappa})=(\kappa-1/d)\omega_{1d}^2$.
\end{theorem}
\begin{proof}
Our goal is  to increase the speed by increasing  the purity in an optimum manner. This is achieved if we let one or more  off-diagonal elements of the density matrix to be increased. The diagonal elements of the density matrix do not play a role in the speed and among the set of all off-diagonal elements, $\varrho_{1d}$ plays a more important role than the others, as its absolute value  appears as the coefficient of the largest squared energy gap $\omega_{1d}^2$ in Eq. \eqref{SqSpeed-Unitary-1}.
Starting from the maximally mixed state $\varrho_0=\frac{1}{d}\sum_{i=1}^{d} \ket{E_i}\bra{E_i}$ which has minimum purity $\kappa=1/d$ and vanishing speed,  and by keeping all the parameters fixed  except $|\varrho_{1d}|$,  we increase the speed by increasing $|\varrho_{1d}|$ under the restriction $|\varrho_{1d}|\le \sqrt{\varrho_{11}\varrho_{dd}}=1/d$.  In turn, the purity is increased to $\kappa=1/d+2|\varrho_{1d}|^2$, or equivalently, the density matrix acquires the new off-diagonal entries  $\varrho _{1d}=\varrho^\ast _{d1}= \e^{i\theta_1}\sqrt{\frac{1}{2}\left(\kappa-1/d\right)}$ for $\kappa\in\left[\frac{1}{d},\frac{1}{d}\left(1+\frac{2}{d}\right)\right]$.
 We arrive therefore at the result asserted by Theorem \ref{Theorem-Opt-k0}  for  the optimal state when  $\kappa\in\left[1/d,\kappa_0\right]$.
\end{proof}
It is interesting to note here that the  state given above is the same as the optimal state  introduced recently in Ref. \cite{ShaoPRR2020} to fulfill the target via the operational definition of QSL.
Our next goal   is to find  the optimum squared speed for the  other values of    purity, namely for $\kappa\in\left[\kappa_0,1\right]$.
For $\kappa\in [\kappa_0,1]$ and dimensions $d\ge 4$, the optimal  state depends in general   on the relative values of the Bohr frequencies  $\omega_{ij}$. Let us define     $\gamma_1=\omega_{1d}^2/\omega_{2,d-1}^2$, $\gamma_2=\omega_{1d}^2/\omega_{3,d-2}^2$,  and so on. Obviously,  $\gamma_i < \gamma_{i+1}$.
Then,  Theorems \ref{Theorem-Opt-Main-1} and \ref{Theorem-Opt-Main-2}  present the optimal states when  $\gamma_1\ge 2$ and $\gamma_1<2$, respectively.
\begin{theorem}\label{Theorem-Opt-Main-1}
When $\gamma_1\ge 2$,  the optimal state is given by
\begin{eqnarray} \nonumber
\varrho_{\kappa\in [\kappa_0,1]}&=& 2\left(\frac{1}{d}+x\right)\ket{\Psi_1}\bra{\Psi_1} \qquad (\textrm{for } \gamma_1\ge 2) \\ \label{rhoOP-1-2}
&+&\left(\frac{1}{d}-\frac{2x}{d-2}\right)\sum_{i=2}^{d-1}\ket{E_i}\bra{E_i},
\end{eqnarray}
for $\kappa\in [\kappa_0,1]$, where  $\ket{\Psi_1}=\frac{1}{\sqrt{2}}(\ket{E_1}+\e^{-i\theta_1}\ket{E_d})$ and
\begin{equation}\label{x-1}
x=\frac{d-2}{2d(d-1)}\left[-1+d\sqrt{\frac{-1+(d-1)\kappa}{d-2}}\right].
\end{equation}
\end{theorem}
\begin{proof}
For a proof see Appendix \ref{Appendix-Theorem-Opt-Main-1}.
\end{proof}

For $\gamma_1<2$,  the calculations to specify the nonzero entries  become more complicated, however, a framework to identify  the optimal state can be developed. In this case, we need to define two purities $\kappa_1$ and $\kappa_2$ as
\begin{eqnarray}\label{kappa1}
\kappa_1&=&\frac{4+(d-2)(2-\gamma_1)^2}{[2(d-1)-(d-2)\gamma_1]^2}, \\ \label{kappa2}
\kappa_2&=&\frac{4+d(2-\gamma_1)^2}{[2(d-1)-(d-2)\gamma_1]^2}.
\end{eqnarray}

\begin{theorem}\label{Theorem-Opt-Main-2}
When $\gamma_1<2$ and depending on the value of $\kappa$,  the optimal state can be obtained as follows.
\begin{itemize}
\item For   $\kappa\in[\kappa_0,\kappa_1]$, the optimal state is still given by Eq. \eqref{rhoOP-1-2}.
\item For $\kappa\in [\kappa_1,\kappa_2]$,  the optimal state is given by
\begin{eqnarray}\nonumber
\varrho_{\kappa\in[\kappa_1,\kappa_2]} &=& 2\left(\frac{1}{d}+x_0\right)\ket{\Psi_1}\bra{\Psi_1}\quad (\textrm{for } \gamma_1< 2) \\  \label{rhoOP-2-1}
&+&\left(\frac{1}{d}-\frac{2x_0}{d-2}\right)\sum_{i=2}^{d-1}\ket{E_i}\bra{E_i} \\ \nonumber
&+&\varrho_{2,d-1}\ket{E_2}\bra{E_{d-1}}+\varrho_{2,d-1}^\ast\ket{E_{d-1}}\bra{E_{2}},
\end{eqnarray}
where $\varrho_{2,d-1}=\e^{i\theta_2}\sqrt{(\kappa-\kappa_1)/2}$ with  $\theta_2$ as an arbitrary phase, and
\begin{equation}\label{x0}
x_0=\frac{1}{d}\left[\frac{\gamma_{1}-1}{2(d-1)/(d-2)-\gamma_{1}}\right].
\end{equation}
\item
To proceed further for $\kappa\ge \kappa_2$ and to reach the maximum purity of unity, we must  repeat steps  similar to those  used to derive Eq. \eqref{rhoOP-2-1}. For more detail, we refer to  step 4 of the proof. The calculations become more complicated due to the need to account for   other threshold ratios  such as $\gamma_2$,  $\gamma_1/\gamma_2$, and so on. However, for dimensions $d\le 4$ the only relevant threshold parameter is $\gamma_1$, as such  the steps given above are exhaustive to obtain  the optimal state. Figure \ref{SchematicOverview} represents a schematic overview of how this framework works.
\end{itemize}
\end{theorem}
\begin{proof}
A sketch of the steps to be followed is provided in  Appendix  \ref{Appendix-Theorem-Opt-Main-2}.
\end{proof}
\begin{figure}[t]
\includegraphics[scale=0.4]{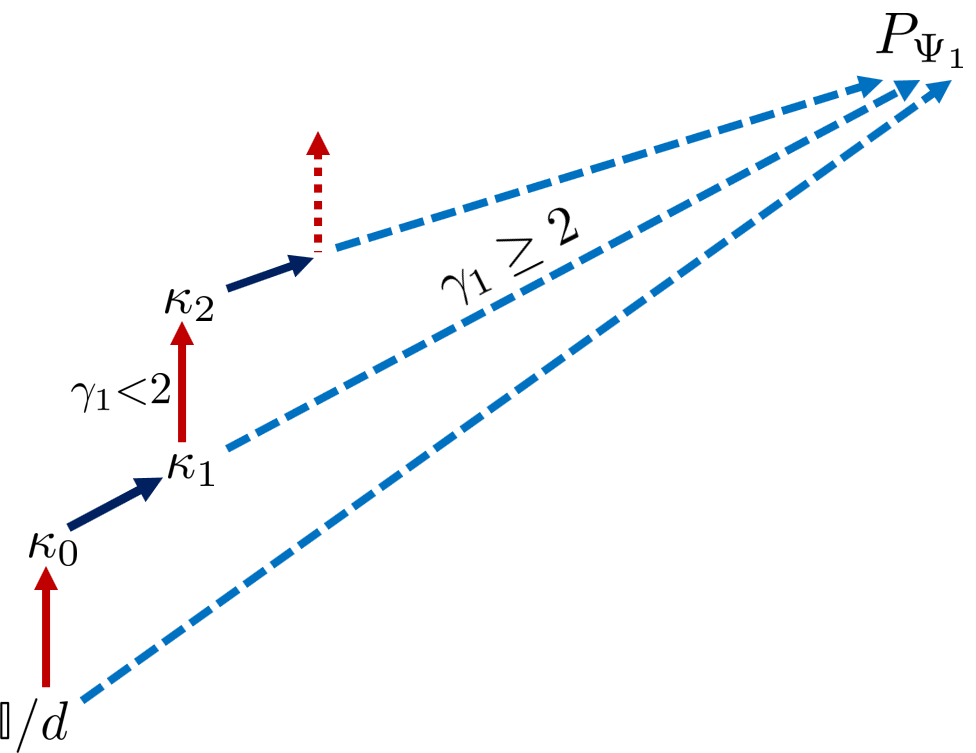}
\caption{A schematic overview  to obtain the optimal state for a given purity $\kappa$, depending on the threshold ratio $\gamma_1$. Starting from the maximally mixed state $\Id/d$, each upward  arrow (red-solid arrow) reduces the rank by $1$, such that $\textrm{rank}\{\rho_{\kappa_0}\}=d-1$, $\textrm{rank}\{\rho_{\kappa_2}\}=d-2$. In each upright arrow (blue-dashed arrow), on the other hand, we do not have any reduction of the rank except at the destination point $P_{\Psi_1}$, at which we have $\textrm{rank}\{P_{\Psi_1}\}=1$  }
\label{SchematicOverview}
\end{figure}


\subsection{Examples: optimal mixed quantum states for   $d\le 4$}
After giving a framework to obtain the optimal states, in this subsection we   present  the optimal states for the first three  simple but important cases $d=2,3$ and $4$. We postpone their proof until we  provide a proof sketch for our main theorems for a general $d$ in Appendix  \ref{Appendix-Theorem-Opt-Main-1,2}.

\subsubsection{Optimal mixed quantum states for   $d=2$}
In this case, as demonstrated in Fig. \ref{SchematicOverview}, the optimal state is achieved  by the first upright arrow, i.e., the path $\Id/d \rightarrow P_{\Psi_1}$.
Actually, since  $\kappa_0=1$ for $d=2$,   Theorem \ref{Theorem-Opt-k0} presents the optimal  squared speed of a qubit    for the whole range  $\kappa\in[1/2,1]$. In this case, Eqs  \eqref{SqSpeed-Unitary-1} and  \eqref{rhoOP-0} give the  optimal speed and its  corresponding optimal state  as  $v^2(\varrho_{\kappa})=(\kappa-1/2)\omega_{12}^2$ and
\begin{equation}\label{Optimum-d=2}
\varrho_{\kappa}=\begin{pmatrix}
\frac{1}{2}& \e^{i\theta_1}\sqrt{\frac{1}{2}\left( \kappa-1/2\right) }\\
\e^{-i\theta_1}\sqrt{\frac{1}{2}\left( \kappa-1/2\right) } & \frac12 \\
\end{pmatrix},
\end{equation}
respectively.
\subsubsection{Optimal mixed quantum states for   $d=3$}
For $d=3$,  we have to take the path $\Id/d\rightarrow \kappa_0\rightarrow P_{\Psi_1}$.
In this case, the   state with optimal speed is given by
\begin{equation}\label{Optimum-d=3}
\varrho_{\kappa}=\begin{pmatrix}
\varrho_{11} & 0 & |\varrho_{13}|\e^{i\theta_1}\\
0 & \varrho_{22} & 0\\
|\varrho_{13}|\e^{-i\theta_1}  & 0 & \varrho_{11}
\end{pmatrix},
\end{equation}
where $\theta_1$ is an arbitrary phase, and matrix entries  $\varrho_{11}$, $|\varrho_{13}|$ and $\varrho_{22}$  are given in Table \ref{Table-d=3}.
\begin{table}[h!]
  \begin{center}
    \caption{Parameters of the optimum state \eqref{Optimum-d=3} for $d=3$.}
    \label{Table-d=3}
    \begin{tabular}{|l|c|c|c|} \hline
       \diagbox{Purity}{State}& $\varrho_{11}$ & $|\varrho_{13}|$ & $\varrho_{22}$ \\       \hline
      $\frac{1}{3}\le \kappa \le \frac{5}{9}$ & $\frac{1}{3}$ & $\sqrt{\frac12\left( \kappa-\frac{1}{3}\right) }$ & $\frac{1}{3}$ \\
      $\frac{5}{9}\le \kappa \le 1$ & $\frac{1+\sqrt{2\kappa-1}}{4}$  & $\frac{1+\sqrt{2\kappa-1}}{4}$ & $\frac{1-\sqrt{2\kappa-1}}{2}$ \\ \hline
    \end{tabular}
  \end{center}
\end{table}
\subsubsection{Optimal mixed quantum states for   $d=4$}
For $d\ge 4$, the optimal  state depends not only on the purity $\kappa$, but also on the relative values of the Bohr frequencies $\omega_{ij}$. In particular for $d=4$, depending on whether $\gamma_1$ is greater than  or less than 2,  we have to take the path $\Id/d \rightarrow \kappa_0\rightarrow \kappa_1\rightarrow P_{\Psi_1}$ or $\Id/d \rightarrow \kappa_0\rightarrow \kappa_1\rightarrow \kappa_2\rightarrow  P_{\Psi_1}$.
 In this case, the optimal state is given by
\begin{equation}\label{Optimum-d=4}
\varrho_{\kappa}=\begin{pmatrix}
\varrho_{11} & 0  & 0  &  |\varrho_{14}|\e^{i\theta_1} \\
 0 & \varrho_{22} & |\varrho_{23}| \e^{i\theta_2} & 0 \\
0 & |\varrho_{23}| \e^{-i\theta_2} & \varrho_{22} & 0 \\
|\varrho_{14}|\e^{-i\theta_1} & 0 & 0 & \varrho_{11} \\
\end{pmatrix},
\end{equation}
where its nonzero entries,  depending on whether  the threshold parameter $\gamma_{1}=\omega^2_{14}/\omega^2_{23}$ is greater than or less than 2,  are  given in Table \ref{Table-d=4}.
\begin{table}[h!]
  \begin{center}
    \caption{Parameters of the optimal  state \eqref{Optimum-d=4} for $d=4$,  where  $\gamma_{1}=\omega_{14}^2/\omega_{23}^2$,  $\kappa_1 = 4(1/4+x_0)^2+2(1/4-x_0)^2$, $\kappa_2=\kappa_1+2(1/4-x_0)^2$, and $x_{0}=(\gamma_{1}-1)/[4(3-\gamma_{1})]$.}
    \label{Table-d=4}
    \begin{tabular}{|c|c|c|c|c|c|}\hline
   \diagbox{Purity}{State} & $\gamma_{1}$ & $\varrho_{11}$ & $|\varrho_{14}|$ & $\varrho_{22}$ & $|\varrho_{23}|$  \\       \hline
      $\frac{1}{4}\le \kappa \le\frac{3}{8}$ & $-$ &  $\frac{1}{4}$ & $\frac{\sqrt{8\kappa-2}}{4}$ & $\frac{1}{4}$ & $0$  \\
      $\frac{3}{8}\le \kappa \le 1 $ & $\ge2$ & $\frac{1+\sqrt{6\kappa-2}}{6}$  & $\frac{1+\sqrt{6\kappa-2}}{6}$ & $\frac{2-\sqrt{6\kappa-2}}{6}$ & $0$  \\
      $\frac{3}{8}\le \kappa \le\kappa_1 $ & $<2$ & $\frac{1+\sqrt{6\kappa-2}}{6}$  & $\frac{1+\sqrt{6\kappa-2}}{6}$ & $\frac{2-\sqrt{6\kappa-2}}{6}$ & $0$   \\
      $\kappa_1 \le \kappa \le \kappa_2 $ & $<2$ & $\frac{1}{4}+x_0$  & $\frac{1}{4}+x_0$ & $\frac{1}{4}-x_0$ & $\sqrt{\frac{1}{2}( \kappa-\kappa_1)}$  \\
      $\kappa_2 \le \kappa \le  1 $ & $<2$ &  $\frac{1+\sqrt{2\kappa-1}}{4}$  & $\frac{1+\sqrt{2\kappa-1}}{4}$ & $\frac{1-\sqrt{2\kappa-1}}{4}$ & $\frac{1-\sqrt{2\kappa-1}}{4}$ \\ \hline
    \end{tabular}
  \end{center}
\end{table}

\section{Simulation and comparison with WY   skew information}\label{Sec-Simulation}
In this section we examine  a numerical simulation performed on the squared speed \eqref{SqSpeed-Unitary-1}  for $d=4$. Our results show, as expected, that  the   optimal state \eqref{Optimum-d=4} is always an upper bound for this speed.
To generate a uniform distribution on the set of all $d$-dimensional quantum states $\varrho$ we follow the method presented in \cite{KarolPRA1998}.
A general $d$-dimensional quantum  state $\varrho$ can be represented by a diagonal density matrix $\varrho_{\diag}=\diag\{p_1,\cdots,p_d\}$ and a unitary matrix $U\in SU(d)$ as $\varrho=U\varrho_{\diag} U^\dag$.
Accordingly, one can  proceed to construct the uniform density matrices  through two steps: First, a uniform diagonal density matrix $\varrho_{\diag}=\diag\{p_1,\cdots,p_d\}$ is generated by using a uniform distribution of  probabilities. In the second route, we use a uniform distribution on  unitary transformations  $U\in SU(d)$, and  generate a uniform  density matrix $\varrho=U\varrho_{\diag} U^\dag$. A simulation of the  squared speed for $d=4$ is given in Fig. \ref{Fig-Simulation}. To construct this figure, we generate one million random quantum states; ten thousand for random diagonal states and one hundred for random unitary matrices. Moreover,  by choosing $\omega_{14}=\sqrt{2}$, we have normalized the optimal speed to unity in the sense that $v^2(\Psi_1)=1$. Our results confirm that the squared speed is bounded from above by the analytical optimal squared speed given by Eq. \eqref{Optimum-d=4} (the red curve).

Now, two comments are in order. (i) The figure shows a sparse distribution of the randomly generated states for large values of  purity. This  comes from the fact that the set of pure states is measure zero, as such, the probability of obtaining random diagonal pure states is zero. In light of this, the probability of generating randomly diagonal states with large purity is very small so that the density of  generated states is reduced  by increasing the purity $\kappa$.  (ii) It follows from the  figure   that   the optimal state (the red curve)  is not reachable randomly, unless  for low purities. This is evident from our result since  the optimal state \eqref{Optimal-X-state}  is a $X$-type state which is of  measure zero among all  quantum states.  This implies  that  for the more complicated case of $\gamma_1<2$ of   high-dimensional systems  where the numerical simulation is inevitable, the probability of finding the optimal state is almost  zero  unless the search is performed    among the  states which  possess  the  $X$-type form \eqref{Optimal-X-state}.

We have to stress here that the quantum speed, as well as the  states with  maximum speed,  depends in general  on the chosen  metric.   Accordingly,  for the metrics other than the Euclidean  metric that we have used in this work,   state \eqref{Optimal-X-state} is not guaranteed to be optimal  for the  whole values of purity. In such cases, however, Eq.  \eqref{Optimal-X-state} can be exploited in order to derive tight lower bounds for the maximum speed. To see this, let us examine  the squared speed based on the Wigner-Yanase (WY) skew information \cite{MondalPLA2016, PiresPRX2016}, which takes the following form in our notation
\begin{eqnarray} \label{SqSpeed-Unitary-2}
v_{\textrm{WY}}^2(\varrho)&=&-\Tr\left[ H,\sqrt{\varrho}\right] ^{2}=  \sum_{i,j} |(\sqrt{\varrho})_{ij}|^2 \omega_{ij}^2.
\end{eqnarray}
This, clearly, is similar in form to the Euclidean-based squared speed \eqref{SqSpeed-Unitary-1}, except that here the off-diagonal  entries of $\sqrt{\varrho}$ are  replaced by the ones of $\varrho$.  Figure \ref{Fig-Simulation-WY}  shows a numerical simulation of this squared speed,  plotted against the purity.  The red curve, on the other hand, shows the squared speed \eqref{SqSpeed-Unitary-2}, obtained for   the state $\varrho_\kappa$ of Eq.  \eqref{Optimal-X-state}.  To obtain the red curve, we have used the  fact that  $\sqrt{\varrho_{\kappa}}$ is again a persymmetric  $X$-state as Eq. \eqref{Optimal-X-state}, however, with the entries $(\sqrt{\varrho_{\kappa}})_{11}=(\sqrt{\lambda_{+}}+\sqrt{\lambda_{-}})/2$ and  $|(\sqrt{\varrho_{\kappa}})_{14}|=(\sqrt{\lambda_{+}}-\sqrt{\lambda_{-}})/2$ where $\lambda_{\pm}=(\varrho_\kappa)_{11}\pm|(\varrho_\kappa)_{14}|$. Same relations hold for entries $(\sqrt{\varrho_{\kappa}})_{22}$ and $|(\sqrt{\varrho_{\kappa}})_{23}|$  in terms of $(\varrho_\kappa)_{22}$ and $|(\varrho_\kappa)_{23}|$.  Clearly, for some regions of low purity, the speed exceeds the speed limit proposed by our  optimal state \eqref{Optimal-X-state}.  This, in turn, implies that the  state \eqref{Optimal-X-state}, which is optimal  by means of the speed defined by Eq. \eqref{SqSpeed-Unitary-1}, no longer remains  optimal in general  for other definitions of speed.  Such optimal state, however, could shed some light on the  QSL bounds.
\begin{figure}[t]
\includegraphics[scale=0.68]{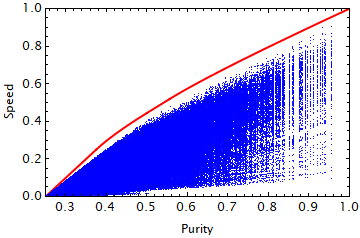}
\caption{A numerical simulation of the Euclidean-based squared speed \eqref{SqSpeed-Unitary-1}  for $d=4$ when  $\gamma_1=3/2$ ($\omega_{14}=\sqrt{2}$, $\omega_{23}=2/\sqrt{3}$). The extent of quantum states in the speed-purity plane are shown by one million randomly generated  states (blue points).   The red curve denotes the squared  speed corresponding to the optimal state  \eqref{Optimum-d=4}. For a given purity, no randomly generated state can have speed higher than the speed of our  optimal  state \eqref{Optimum-d=4}.    }
\label{Fig-Simulation}
\end{figure}
\begin{figure}[t]
\includegraphics[scale=0.68]{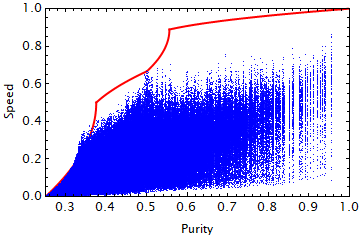}
\caption{A numerical simulation of the WY-based squared speed \eqref{SqSpeed-Unitary-2}  for $d=4$. The parameters are the same as in Fig \ref{Fig-Simulation}.  The red curve denotes the WY-squared  speed corresponding to  the proposed optimal  state  \eqref{Optimum-d=4}. The  state \eqref{Optimum-d=4}, which is optimal for the speed defined by Eq. \eqref{SqSpeed-Unitary-1}, no longer remains  optimal for the speed defined by Eq. \eqref{SqSpeed-Unitary-2}, especially when purity is sufficiently low. }
\label{Fig-Simulation-WY}
\end{figure}

\section{Discussion: Entanglement and coherence}\label{Sec-Discussion}
It is understood that quantum resources such as entanglement and coherence could  speed up quantum evolutions and in particular for unitary evolutions  it is shown that entanglement can reduce the quantum speed limit time \cite{GiovannettiPRA2003,GiovannettiEU2003,BorrasPRA2006}.  In this section we discuss  the role played by these resources.

\subsection{Entanglement}\label{SubSec-Entanglement}
To study the effect of  entanglement on the optimal speed,  suppose the $d$-dimensional system $\mathcal{H}$ constitutes a bipartite system, i.e. $\mathcal{H}=\mathcal{H}_A\otimes \mathcal{H}_B$ so that $d=d_1d_2$ where $d_1$ and $d_2$ are dimensions  of $\mathcal{H}_A$ and $\mathcal{H}_B$, respectively. We assume that the Hamiltonian $H$ is  diagonal in the product basis so  the energy eigenstates provide a product basis, i.e.,
\begin{equation}
\{\ket{E_1}=\ket{00},\ket{E_2}=\ket{01}, \cdots, \ket{E_d}=\ket{d_1-1,d_2-1}\}.
\end{equation} This includes the Hamiltonians of the form \cite{BatlePRA2005}
\begin{equation}
H=H_A\otimes \Id_B+\Id_A\otimes H_B,
\end{equation}
where $H_A$ and $H_B$ are local Hamiltonians of the subsystems  $\mathcal{H}_A$ and $\mathcal{H}_B$, respectively.  To proceed further we first consider   $d=4$,  i.e., a two-qubit system which is  the simplest bipartite system described by a four-level model.

In  this case  the Hamiltonian $H$ is diagonal in the product basis $\ket{E_1}=\ket{00}$, $\ket{E_2}=\ket{01}$, $\ket{E_3}=\ket{10}$, and $\ket{E_4}=\ket{11}$.
To quantify entanglement we use concurrence \cite{WoottersPRL1998} as a measure of entanglement which has a closed formula for two-qubit states.
With the above assumption,  the concurrence  of the optimal state \eqref{Optimum-d=4} is given by
\begin{equation}
\mathcal{C}(\varrho)=\max\{2(|\varrho_{14}|-\varrho_{22}),0\},
\end{equation}
which interestingly is exactly equal to the negativity measure of entanglement defined by $\mathcal{N}(\varrho)=|\varrho^{\textrm{T}_1}|_{1}-1$, where $\varrho^{\textrm{T}_1}$ is a matrix defined by  partial transposing \cite{PeresPRL1996,HorodeckiPLA1996} the quantum state $\varrho$ with respect to the first part, and $\|A\|_{\textrm{tr}}=\Tr\sqrt{A^\dagger A}$ is the trace norm.
Figures \ref{Fig-Concurrence-3}  and \ref{Fig-Concurrence-4} illustrate the concurrence (below-dashed curve) of the optimal state $\varrho_{\kappa}$ in terms of the purity $\kappa$ for two regimes $\gamma_1\ge  2$ and $\gamma_1 < 2$, respectively.
It follows from these  figures that in both cases the optimal state is disentangled for  $\kappa\in[1/4,3/8]$.
{On the other hand, for $\kappa> 3/8$ and in  both regimes,  the optimal state is no longer disentangled; its concurrence monotonically increases  with purity.

For an arbitrary $d$-dimensional composite system, with $d=d_1d_2$,  we first show  that  the optimal state $\varrho_{\kappa}$ is separable  for all purities $\kappa\in\left[1/d,\kappa_0\right]$.
To see this  one can easily check  that the state given by Eq. \eqref{rhoOP-0} possesses  the following representation  in terms of  product states
\begin{eqnarray}\label{rhoOP-0-SEP}
\varrho_{\kappa\in[\frac{1}{d},\kappa_0]}&=&\sum_{i=0}^{d_1-1} \sum_{j=0}^{d_2-1} \left(1/d-\Delta_{ij}\right)\ket{ij}\bra{ij} \\ \nonumber
&+& |\varrho_{1d}|\left[P_{x}^{+}\otimes Q_{x}^{+}+P_{x}^{-}\otimes Q_{x}^{-}\right. \\ \nonumber
&& \;\quad +\left.P_{y}^{+}\otimes Q_{y}^{-}+P_{y}^{-}\otimes Q_{y}^{+}\right].
\end{eqnarray}
Here, the coefficients $\Delta_{ij}$s are zero except for $\Delta_{0,0}=\Delta_{0,d_2-1}=\Delta_{d_1-1,0}=\Delta_{d_1-1,d_2-1}=|\varrho_{1,d}|$, and
 $P_{x}^{\pm}$ and $P_{y}^{\pm}$ are projections on pure states
\begin{eqnarray}
\ket{\pm x }=\frac{1}{\sqrt{2}}(\ket{0}\pm \e^{-i\theta_1/2}\ket{d_1-1}), \\ \nonumber
\ket{\pm y }=\frac{1}{\sqrt{2}}(\ket{0}\pm i\e^{-i\theta_1/2}\ket{d_1-1}),
\end{eqnarray}
of the first subsystem, and $Q_{x}^{\pm}$ and $Q_{y}^{\pm}$ are defined similarly for the second subsystem.

Remarkably, this  separability of a bipartite state for  low purities  is not surprising as
all states in the sufficiently small neighborhood of the maximally mixed state are necessarily  separable. More precisely,    any  bipartite $d_1\times d_2$  mixed state with purity less than $1/(d-1)$ has positive partial transposition  \cite{KarolPRA1998}.  Our result   shows however a wider range of separability for  the optimal state;  the optimal state is separable  for purities less than $\kappa_0=1/d+2/d^2$.

For $\kappa>\kappa_0$, on the other hand, the optimal state for  an  arbitrary $d$-dimensional composite system  is always entangled. To see this,  let us  first consider the case  $\gamma_1\ge 2$ for which the optimal state $\varrho_{\kappa\in [\kappa_0,1]}$ is given by Eq. \eqref{rhoOP-1-2}. Simple calculation shows that the partial transposed matrix $\varrho^{\textrm{T}_1}_{\kappa\in [\kappa_0,1]}$ has  $(d-1)$ positive eigenvalues, namely  $1/d-2x/(d-2)$, $1/d+x$ and $2/d+x(d-4)/(d-2)$ with multiplicities $(d-4)$, $2$ and $1$ respectively, and one negative eigenvalue equal to $\lambda_{-}=-xd/(d-2)$. These eigenvalues are independent of how the composite system $\mathcal{H}$ is partitioned into  two subsystems $\mathcal{H}_A$ and $\mathcal{H}_B$.  Note that $x$, defined by Eq. \eqref{x-1},  is a monotonically increasing function of $\kappa$  and ranges from $0$ to $(d-2)/(2d)$ as $\kappa$ goes  from $\kappa_0$ to $1$.    Accordingly,  negativity is given by $\mathcal{N}(\varrho_{\kappa})=2|\lambda_{-}|=2xd/(d-2)$ which   is monotonically increasing in interval  $\kappa\in [\kappa_0,1]$ and reaches the maximum value of unity at $\kappa=1$.

In  regime $\gamma_1<2$ and $d>4$ for  which  more  threshold ratios  are needed to characterize the optimal state, calculation of the negativity   becomes more complicated for  $\kappa\in[\kappa_0,1]$. However, as we mentioned in Theorem \ref{Theorem-Opt-Main-2}, the optimal state is still given by Eq. \eqref{rhoOP-1-2}  for $\kappa\in[\kappa_0,\kappa_1]$, so that its  negativity is still given by $\mathcal{N}(\varrho_{\kappa\in[\kappa_0,\kappa_1]})=2xd/(d-2)$.    In view of this, the optimal state starts to be entangled at purity $\kappa_0=1/d+2/d^2$ for arbitrary $d$. This entanglement increases monotonically in the interval  $\kappa\in [\kappa_0,\kappa_1]$ and reaches  the maximum value  $\mathcal{N}(\varrho_{\kappa_1})=2x_0d/(d-2)$,  where $x_0$ is defined by Eq. \eqref{x0} (see Fig. \ref{Fig-Concurrence-4} when  $d=4$).

For $\kappa\in[\kappa_1,\kappa_2]$ the optimal state  is given by Eq. \eqref{rhoOP-2-1}, and can be rewritten as
\begin{eqnarray}\nonumber
\varrho_{\kappa\in[\kappa_1,\kappa_2]} = 2\left(1/d+x_0\right)\ket{\Psi_1}\bra{\Psi_1}+\sigma^{\textrm{sep}}_{\kappa\in[\kappa_1,\kappa_2]} \;\; (\textrm{for } \gamma_1< 2),
\end{eqnarray}
where $\sigma^{\textrm{sep}}_{\kappa}$ is a separable (unnormalized) state with a  product representation similar to that given by  Eq. \eqref{rhoOP-0-SEP}. The spectrum of the partial transposed matrix $\varrho^{\textrm{T}_1}_{\kappa\in[\kappa_1,\kappa_2]}$ has one negative eigenvalue $-x_0d/(d-2)$ leading therefore to the negativity $\mathcal{N}(\varrho_{\kappa\in[\kappa_1,\kappa_2]})=2x_0d/(d-2)$ which is equal  to $\mathcal{N}(\varrho_{\kappa_1})$.    This explains  why the entanglement of the optimal state  is constant  in interval $\kappa\in[\kappa_1,\kappa_2]$. This procedure can  be continued until the purity tends to $\kappa=1$,  corresponding to the optimal pure state $P_{\Psi_1}=\ket{\Psi_{1}}\bra{\Psi_1}$ which has the maximum value of unity for  the negativity.

\begin{figure}[t]
\includegraphics[scale=0.75]{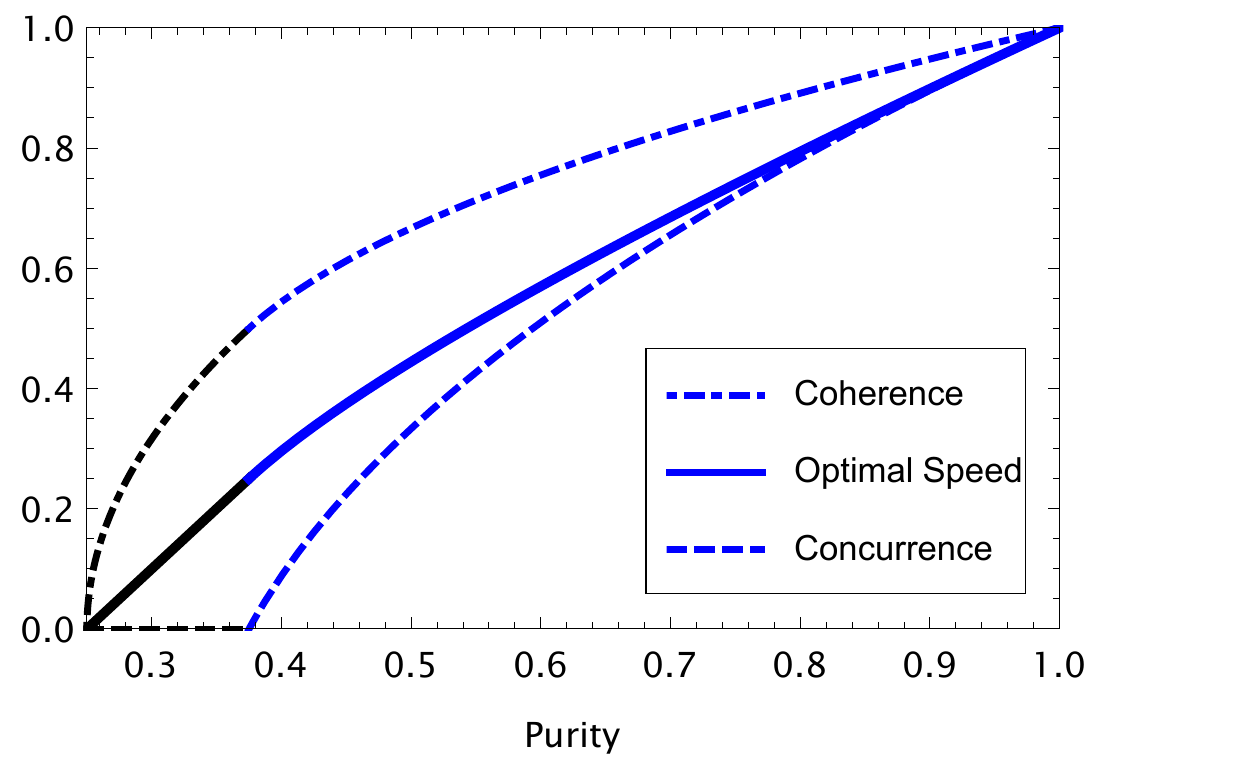}
\caption{Plots of coherence (above-dotdashed curve), squared optimal speed (middle-solid curve), and concurrence (below-dashed curve) in terms of purity $\kappa$ when $\gamma_1=5/2$ ($\omega_{14}=\sqrt{2}$, $\omega_{23}=2/\sqrt{5}$).
The colors correspond to different purity regions: black for $[1/d,\kappa_0]$, and blue for $[\kappa_0,1]$. }
\label{Fig-Concurrence-3}
\end{figure}
\begin{figure}[t]
\includegraphics[scale=0.75]{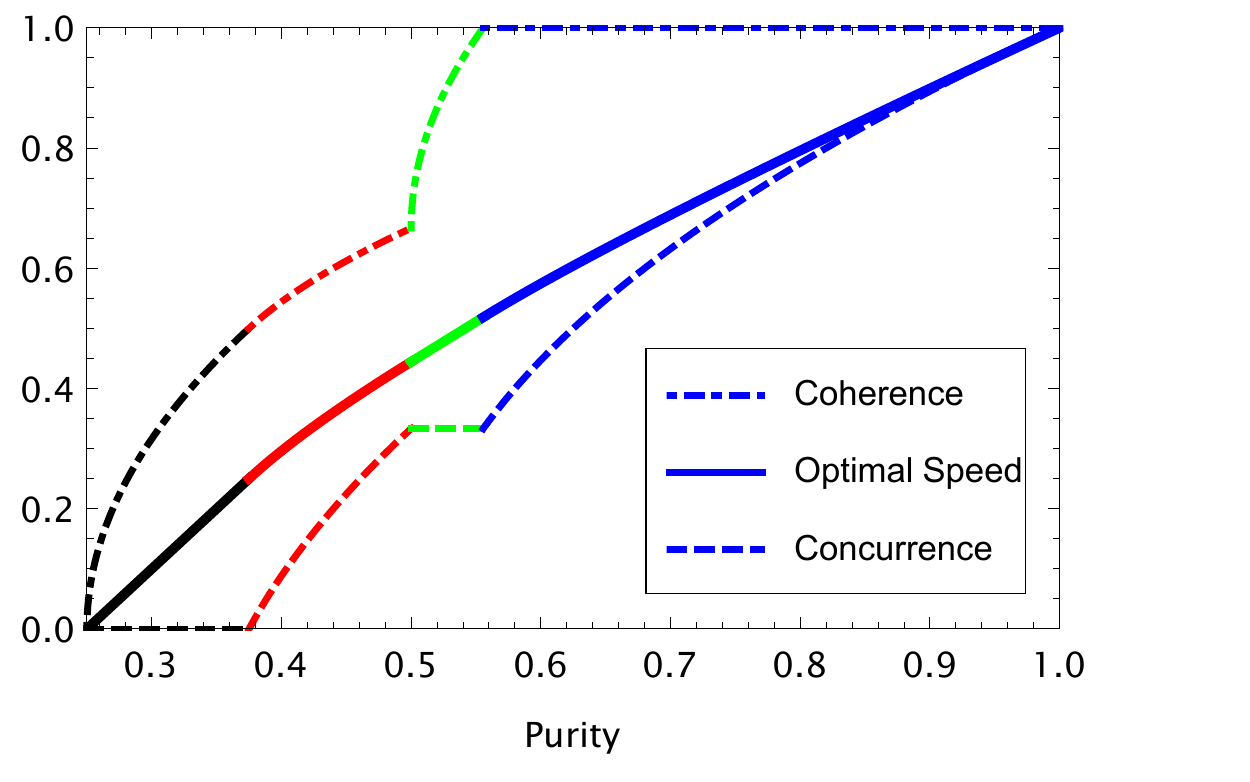}
\caption{Plots of coherence (above-dotdashed curve), squared optimal speed (middle-solid curve), and concurrence (below-dashed curve) in terms of purity $\kappa$ when $\gamma_1=3/2$ ($\omega_{14}=\sqrt{2}$, $\omega_{23}=2/\sqrt{3}$).
The colors correspond to different purity regions: black for $[1/d,\kappa_0]$, red for $[\kappa_0,\kappa_1]$, green for $[\kappa_1,\kappa_2]$, and blue for $[\kappa_2,1]$. }
\label{Fig-Concurrence-4}
\end{figure}

\subsection{Coherence}\label{SubSec-Coherence}
According to the Eq. \eqref{SqSpeed-Unitary-1}, the squared speed of a system evolved unitarily under the Hamiltonian $H$ is nonzero if and only if the system has nonzero coherence in the energy eigenbasis. However, despite the key role of coherence in the evolution speed of the system, not every increase in the coherence  will result in an increase in the speed.   We are also interested in  to find out how the optimal speed   is related to the coherence of the optimal state. To this aim, we use the $l_1$-norm of coherence \cite{PlenioPRL2014}
\begin{eqnarray}\label{QC-l1-norm}
C_{l_1}(\varrho)=\sum_{i\ne j}|\varrho_{ij}|,
\end{eqnarray}
to quantify the coherence of $\varrho$  in the Hamiltonian basis.

Figure \ref{Fig-Concurrence-Sim} shows the numerical simulations of the  squared speed plotted against the coherence for $d=4$ when $\gamma_1<2$. The simulations are  conducted for four  values of purity, corresponding to the four regions $[1/4,\kappa_0]$,  $[\kappa_0,\kappa_1]$, $[\kappa_1,\kappa_2]$, and $[\kappa_2,1]$. For comparison, the optimal state associated with  each purity is also drawn by a red point.
It follows from these  figures that the state with maximum coherence does not always result in the maximum speed. Moreover, the coherence of the optimal state  is not even as large as the mean coherence of the randomly generated states.
\begin{figure}
\includegraphics[scale=0.33]{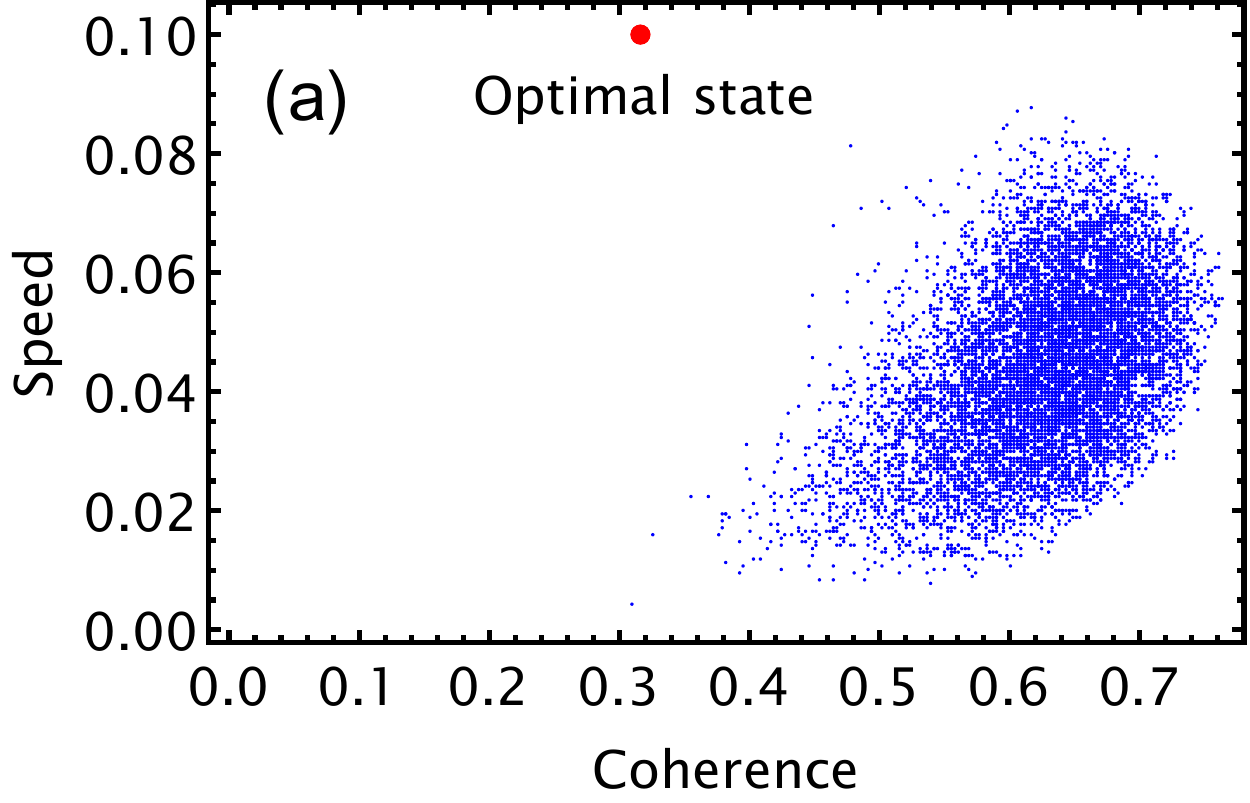}\hfill
\includegraphics[scale=0.33]{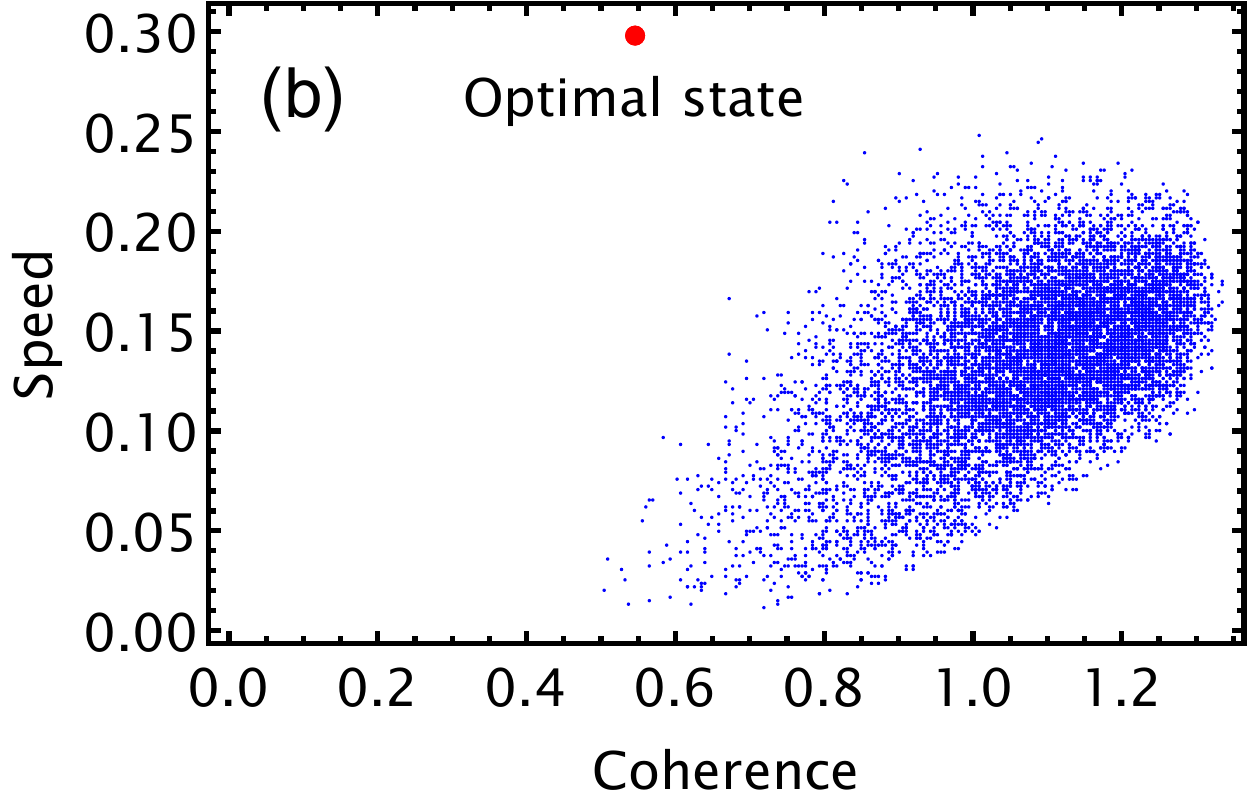}
\\[\bigskipamount]
\includegraphics[scale=0.33]{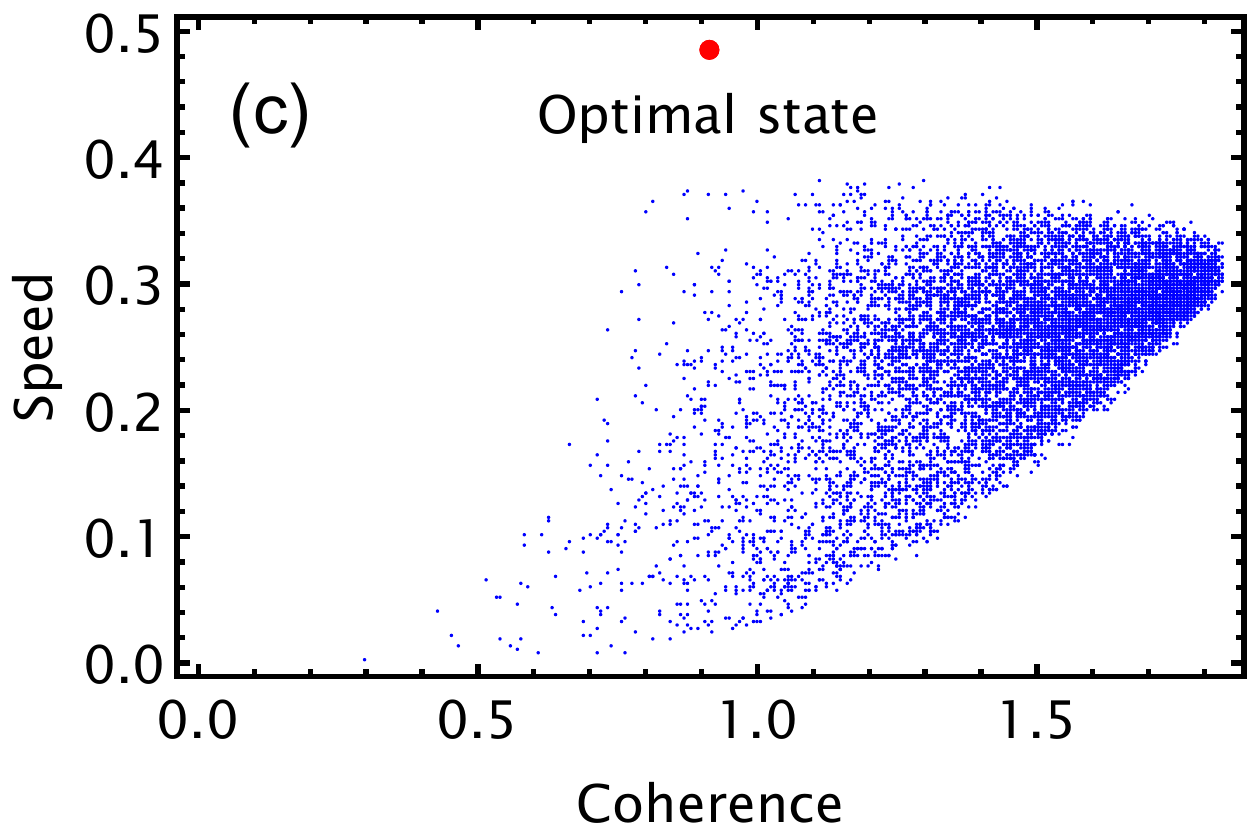}\hfill
\includegraphics[scale=0.33]{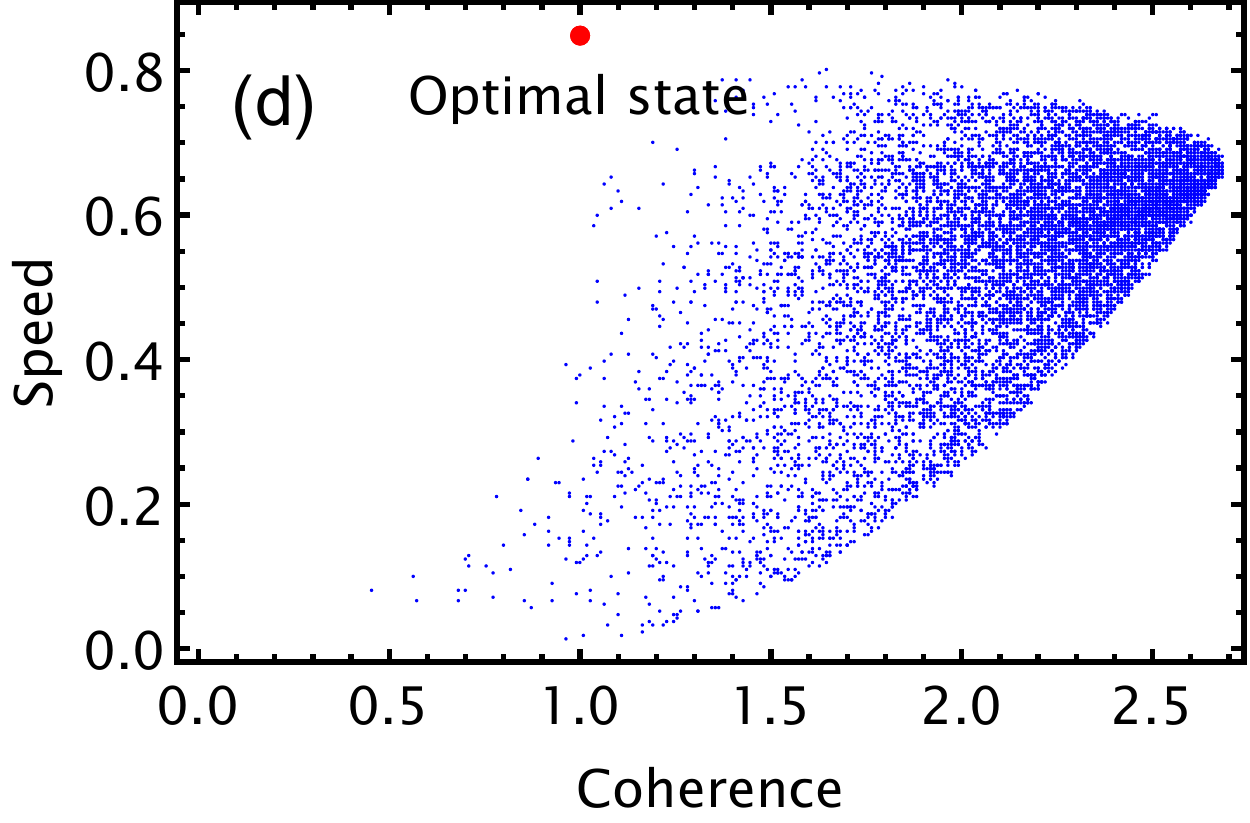}
\caption{Plots of simulations of the squared speed against coherence for $\gamma_1=3/2$ and  different purities: (a) $\kappa=0.3$, (b) $\kappa=0.4$, (c) $\kappa=0.53$, and (d) $\kappa=0.85$. In each plot, the extent of quantum states in the speed-coherence plane is shown by 10000 randomly generated states (blue points). The red points represent  the optimal states.  }
\label{Fig-Concurrence-Sim}
\end{figure}

For the optimal state, given by Eq. \eqref{Optimum-d=4} for $d=4$,  the $l_1$-norm of coherence reduces simply to $C_{l_1}(\varrho)=2(|\rho_{14}|+|\rho_{23}|)$.
Since the optimal states possess the $X$-type symmetry, only the coherence caused by the secondary diameter entries $\varrho_{14}$ and $\varrho_{23}$  have a nonvanishing  effect on the optimal speed,  as such,  the maximum value of the  coherence could  not exceed unity.
Figures \ref{Fig-Concurrence-3} and \ref{Fig-Concurrence-4} illustrate  the coherence, the speed, and the concurrence of the optimal state $\varrho_{\kappa}$  in terms of the purity $\kappa$ for two regimes $\gamma_1\ge 2$ and $\gamma_1< 2$, respectively.
These figures depict the optimal speed as a monotonic function of both quantities  $\mathcal{C}(\varrho_{\kappa})$ and $C_{l_1}(\varrho_{\kappa})$, however, the converse is not  true in general as we already mentioned it  in Figure \ref{Fig-Concurrence-Sim} for the coherence.

A comparison of  Figures \ref{Fig-Concurrence-3} and \ref{Fig-Concurrence-4} shows that when  $\kappa \in [1/4, \kappa_{1}]$, all  mentioned quantities are independent of the value of $\gamma_{1}$ as we expect from Table \ref{Table-d=4}.  When moving away from the maximally mixed state $\Id/4$, namely for $\kappa \in [ \kappa_{1},1]$, their behaviors  become different, depending   on  whether the threshold ratio  $\gamma_1$ is greater than or less than 2.
In particular, for  $\gamma_{1}< 2$,  both coherence and concurrence  have sudden changes at two purities $\kappa_1$ and $\kappa_2$.  At $\kappa=\kappa_1$, the  coherence increasing is suddenly  accelerated,  while the concurrence  increasing is suddenly stopped.  For $\kappa\in[\kappa_1,\kappa_2]$,   any increase in purity is consumed to increase the coherence to its maximum value of unity, while the concurrence maintains its constant value $[2\sqrt{6\kappa_1-2}-1]/6$. At $\kappa=\kappa_2$, on the other hand,  the  coherence increasing is suddenly  stopped after reaching  its maximum value of unity,  while the concurrence  starts to  increase. For the remaining range of purity, i.e., for $\kappa\in[\kappa_2,1]$,  their behavior becomes opposite  in the sense that while coherence remains unchanged, the concurrence is increased in order to reach its maximum value of unity.
These  properties can have applications, for instance,  in the real dynamics where the system-environment interaction is inevitable and the system  loses its purity. In such cases, the optimal-speed dynamics can be  utilized to  reduce  the operation time in order  to maintain  the quantum features such as quantum coherence and quantum entanglement.

Turning our attention to the definition of squared speed given by Eq. \eqref{SqSpeed-Unitary-1}, one can interpret it as  a competition between the energy gaps $\omega_{ij}$ and the off-diagonal entries  $|\varrho_{ij}|$   which have direct contribution to the quantum coherence (see Eq. \eqref{QC-l1-norm}). In view  of the above discussion,    the role of energy gaps is dominant  over the off-diagonal entries when $\gamma_1\ge 2$, however, for $\gamma_1<2$ the roles are reversed, i.e., the coherence plays a greater role than the energy gaps. Using the general definition of $l_1$-norm of coherence,  Eq.  \eqref{QC-l1-norm}, a similar discussion  can be made for the  coherence of the $d$-dimensional optimal state \eqref{Optimal-X-state}.

\section{Conclusion}\label{Sec-Conclusion}
In this paper we have considered the question that how fast can a quantum state  evolve or,  more precisely,  among the set of all quantum states with a given definite   purity $\kappa$,  which  initial state represents the optimal one. Based on the notion of squared speed  defined on the  Euclidean metric \cite{BrodyPRR2019},  an analytical  framework to obtain the optimal speed of a  $d$-dimensional system with unitary evolution generated by  a  time-independent Hamiltonian is presented. The  framework works  step-by-step in the sense that starting from the maximally mixed state having minimum purity and zero speed, the purity is increased optimally  in favor of speed.

In any dimension $d$,  and for a given purity $\kappa$, the fastest state is a $X$-state with an additional property  that it is also symmetric with respect to the secondary diagonal.
The off-diagonal entry $\varrho_{1d}$  has nonzero contribution to all  optimal states with arbitrary purity $\kappa\in(1/d,1]$.
Other off-diagonal entries $\varrho_{i,d-i+1}$, however,  can have a nonzero contribution to the optimal speed depending on the ratio of the  Bohr frequencies.
In particular for   $\gamma_1\ge 2$, where $\gamma_{1}=\omega_{1d}^2/\omega_{2,d-1}^2$,   all entries $\varrho_{i,d-i+1}$ vanish  except $\varrho_{1d}$. In this  regime, we have found a complete description  of the optimal state for an arbitrary dimension $d$.

On the other hand  when  $\gamma_1<2$,  the off-diagonal entry $\varrho_{2,d-1}$  makes a  nonzero contribution to the  optimal states  for all  purities   $\kappa\in(\kappa_1,1]$ where $\kappa_1$ is given by Eq. \eqref{kappa1}.   In this regime,  we have presented  a complete solution for $d\le 4$.  For higher dimensional systems, however,   more threshold parameters  are  needed to decide  whether  or not the other off-diagonal entries   have  nonzero contribution in the optimal state, so that a general treatment to find the optimal state  is more complicated.

Entanglement and coherence of the optimal states and their relations with purity  $\kappa$  have been also investigated. We have found  that for the optimal states both resources are monotonic functions of purity.
Although a nonzero coherence, in the energy eigenbasis,  is a necessary and sufficient condition for the nonzero speed, any  increase in the coherence, however,  will not results in an increase in the speed. For the sufficiently low purities, namely for $\kappa \in [1/d, \kappa_{1}]$, the speed is independent of the value of $\gamma_1$, as such, both quantum  coherence and  quantum entanglement  are independent of the value of $\gamma_{1}$.  For $\kappa \in [ \kappa_{1},1]$, on the other hand,  their behaviors  become different depending   on  whether the threshold parameter $\gamma_1$ is greater than or less than 2.
In particular, both coherence and concurrence  display sudden changes in their behaviors when $\gamma_1 <2$, but it is not the case  for $\gamma_1\ge 2$.  Moreover, our  results show that    the role of energy gaps is dominant  over the off-diagonal entries when $\gamma_1\ge 2$, however, for $\gamma_1<2$  the coherence plays the dominant role.

These results could  find applications  in real dynamics where the system-environment interaction is inevitable and the system  loses its purity. In such cases, the optimal-speed dynamics can be  utilized to  reduce  the operation time in order  to maintain   quantum features such as quantum coherence and quantum entanglement.
Having the states  with maximum speed for  each purity $\kappa\in[1/d,1]$,   they can be used to construct a tight QSL.
Due to the widespread use of the QSL  in diverse branches of quantum technology, it is hoped that our results should shed  light on
the several applications including   quantum communication, quantum computation, quantum metrology,
and especially quantum  control.

In open dynamics the purity is not preserved in general, so the generalization of our method  to open dynamics is not straightforward.  In such systems, the radial component of the speed will obtain a nonvanishing contribution to the speed,  leading therefore to the change of state purity. This component, which is expressed in terms of Lindblad operators, vanishes for a unitary evolution \cite{BrodyPRR2019}, so the purity remains constant. To obtain the fastest states in an open system, both radial and tangential components of the speed must be optimized simultaneously. This problem  can be considered for future works.

\section*{acknowledgment}
This work was supported by Ferdowsi University of Mashhad under Grant No.  3/55339 (1400/06/28).

\appendix
\section{Proof of Proposition \ref{Prop-Xstate}}\label{Appendix-Prop-Xstate}
The following two lemmas are essential in proving Proposition \ref{Prop-Xstate}.
\begin{lemma}\label{Lemma-Rho1d=Rhojd}
Suppose $\varrho$ is a  $d$-dimensional density matrix, and $M(1,d)$ is  its 2-dimensional minor obtained by removing all rows and columns except rows and columns $1,d$. If     $\varrho_{11}=\varrho_{dd}=\varrho_{1d}e^{-i\theta}\ne 0$, then    $\varrho_{jd}=\varrho^\ast_{1j}e^{i\theta}$ for $1< j < d$.  Moreover,  vector $\ket{\chi}=(1,0,\cdots,0,-e^{-i\theta})^{\T}$ belongs to $\textrm{Ker}\{\varrho\}$, i.e.  $\varrho\ket{\chi}=0$. Similarly,  for all $2$-order minors $M(i_1,i_2)$,  $1\le i_1 < i_2 \le d$, for which   $\varrho_{i_1i_1}=\varrho_{i_2i_2}=|\varrho_{i_1i_2}|\ne 0$, then  $|\varrho_{ji_1}|=|\varrho_{ji_2}|$ for $j\ne i_1,i_2$.
\end{lemma}
\begin{proof}
The positivity of $\varrho$ requires that all $k$-order ($1\le k\le d$) principal minors of $\varrho$ be nonnegative. Consider the $3$-order principal minors obtained by removing all rows and columns except rows and columns $1,j,d$ for a fixed value of   $j\in [2,\cdots, d-1]$. One can easily show that the determinant of this minor equals $D_3(1,j,d)=-\varrho_{11}|\varrho_{1j}-\varrho_{jd}^\ast\e^{i\theta}|^2$,   which  is negative unless $\varrho_{1j}=\varrho_{jd}^\ast\e^{i\theta}$.
\end{proof}

\begin{lemma}\label{Lemma-Rho1j=zero}
Suppose   $\varrho$  be a state with purity $\kappa$ such that  $\varrho_{11}=\varrho_{dd}=|\varrho_{1d}|\ne 0$. If  $\varrho$ be a state with optimal speed, then  $\varrho_{1j}=0$ for $j=2,\cdots,d-1$. In this case,   $\ket{\Psi_1}=\frac{1}{\sqrt{2}}(\ket{E_1}+\e^{-i\theta}\ket{E_d})$ and $\ket{\Psi_1^{\perp}}=\frac{1}{\sqrt{2}}(\ket{E_1}-\e^{-i\theta}\ket{E_d})$ belongs to the support and kernel of $\varrho$, respectively, in the sense that  $\varrho\ket{\Psi_1}=2\varrho_{11}\ket{\Psi_1}$ and $\varrho\ket{\Psi_1^\perp}=0$.
This result is general in the sense that if  $\varrho_{ii}=\varrho_{d-i+1,d-i+1}=|\varrho_{i,d-i+1}|\ne 0$ for some $i\in \{1,\cdots,\lfloor d/2\rfloor\}$, and $\varrho$ be a state with optimal speed, then  $\varrho_{ij}=0$  for   $j=i+1,i+2, \cdots,d-i$. In this case,   $\ket{\Psi_i}=\frac{1}{\sqrt{2}}(\ket{E_i}+\e^{-i\theta_i}\ket{E_{d-i+1}})$ and $\ket{\Psi_i^{\perp}}=\frac{1}{\sqrt{2}}(\ket{E_i}-\e^{-i\theta_i}\ket{E_{d-i+1}})$ belongs to the support and kernel of $\varrho$, respectively, in the sense that  $\varrho\ket{\Psi_i}=2\varrho_{ii}\ket{\Psi_i}$ and $\varrho\ket{\Psi_i^\perp}=0$.
\end{lemma}
\begin{proof}
We provide a proof for the  case of $i=1$; the proof for an arbitrary $i$ is similar. First note that the case $\varrho_{11}=0$ gives immediately $\varrho_{1j}=0$, a result obtained by nonnegativity of the determinant of $M(1,j)$, so we proceed with the assumption that $\varrho_{11}\ne 0$.
From Eq.  \eqref{SqSpeed-Unitary-1} and applying  Lemma \ref{Lemma-Rho1d=Rhojd}, the squared speed of $\varrho$ is given by
\begin{eqnarray}
v^2(\varrho)=2\left[\varrho_{11}^2\omega_{1d}^2+\sum_{j=2}^{d-1}|\varrho_{1j}|^2\Omega^2_{j}+\sum_{j=3}^{d-1}\sum_{i=2}^{j-1}|\varrho_{ij}|^2\omega^2_{ij}\right],
\end{eqnarray}
where $\Omega^2_{j}=\omega_{1j}^2+\omega_{jd}^2$ for $j=2,\cdots,d-1$.
The optimum squared speed is obtained by maximizing the above speed subject to  the purity constraint $\kappa-\Tr\varrho^2=0$ where
\begin{eqnarray}
\Tr\varrho^2=4\varrho_{11}^2+\sum_{i=2}^{d-1}\varrho_{ii}^2+4\sum_{j=2}^{d-1}|\varrho_{1j}|^2+2\sum_{j=3}^{d-1}\sum_{i=2}^{j-1}|\varrho_{ij}|^2.
\end{eqnarray}
Note that the squared speed is not a function of the diagonal elements $\varrho_{ii}$, except for $\varrho_{11}$  because of its relation to the off-diagonal element $\varrho_{1d}=\varrho_{11}\e^{i\theta}$. Therefore,   the normalization condition $\sum_{i=1}^d\varrho_{ii}=1$ is not relevant  as an equality condition, however, we should check the inequality $0\le \varrho_{11}\le 1/2$ together with the positivity condition $\varrho\ge 0$ at the end of our calculations.

Using the method of Lagrange multiplier, we have to maximize $v^2(\varrho)$ subject to the purity condition $\kappa-\Tr\varrho^2=0$. Denoting by $\mu$ the corresponding Lagrange multiplier, and writing ${\tilde{v}}^2(\varrho)=v^2(\varrho)+\mu (\kappa-\Tr{\varrho^2})$,    we have the following relations from $\partial {\tilde{v}}^2(\varrho)/\partial \varrho_{11}=0$, $\partial {\tilde{v}}^2(\varrho)/\partial |\varrho_{1j}|=0$ ($2\le j\le d-1$), and $\partial {\tilde{v}}^2(\varrho)/\partial |\varrho_{ij}|=0$ ($2\le i<j\le d-1$), respectively
\begin{eqnarray}\nonumber
\varrho_{11}(2\mu-\omega_{1d}^2)=0,  \;  |\varrho_{1j}|(2\mu-\Omega^2_{j})=0,  \; |\varrho_{ij}|(\mu-\omega^2_{ij})=0.
\end{eqnarray}
Since $\varrho_{11}\ne 0$,  the first relation determines the Lagrange multiplier as $\mu=\omega_{1d}^2/2$. Accordingly, for the second relation,  $\mu=\Omega^2_{j}/2$ is not an acceptable solution unless $\omega_{1d}=\Omega_{j}$, which is not possible as  $\Omega^2_{j}=\omega_{1d}^2-2\omega_{1j}\omega_{jd}<\omega_{1d}^2$; as such,  we provide a prove for the assertion that  for the optimal  state $\varrho_{1j}=0$ for $j=2,\cdots,d-1$. For the third relation, $\varrho_{ij}\ne 0$ can be an acceptable solution in some particular cases, but we postpone its discussion while we provide  proofs for   Theorems \ref{Theorem-Opt-Main-1} and \ref{Theorem-Opt-Main-2}.
\end{proof}
We are now in the position to proof the Proposition \ref{Prop-Xstate}.
That the optimality  requires the state possess the  $X$-type symmetry is  a consequence of the successive application of the Lemmas \ref{Lemma-Rho1d=Rhojd} and \ref{Lemma-Rho1j=zero}.  The additional property of being symmetric with respect to the secondary diagonal, i.e.,  $\varrho_{ij}=\varrho_{d-j+1,d-i+1}$ for all $i,j$, comes from the Step 1 below.

\section{Proofs  of   Theorems \ref{Theorem-Opt-Main-1} and \ref{Theorem-Opt-Main-2}}\label{Appendix-Theorem-Opt-Main-1,2}
\subsection{Proof of  Theorem \ref{Theorem-Opt-Main-1}}\label{Appendix-Theorem-Opt-Main-1}
To prove  Theorems \ref{Theorem-Opt-Main-1}, we follow the following steps.

\subsection*{Step 1 ($\kappa \ge \kappa_0$)}
We start from  Eq.  \eqref{rhoOP-0} when  $\kappa$ reaches  its maximum value $\kappa_0$, i.e.,
\begin{eqnarray}
\varrho_{{\kappa_0}}= \frac{2}{d}\ket{\Psi_1}\bra{\Psi_1}
+\frac{1}{d}\sum_{i=2}^{d-1}\ket{E_i}\bra{E_i},
\end{eqnarray}
where $\ket{\Psi_1}=\frac{1}{\sqrt{2}}(\ket{E_1}+\e^{-i\theta_1}\ket{E_d})$ is a state with maximum squared speed.
Clearly $\varrho_{\kappa_0}$  no longer contains $\ket{\Psi_1^{\perp}}=\frac{1}{\sqrt{2}}(\ket{E_1}-\e^{-i\theta_1}\ket{E_d})$ in its range, as such  $\textrm{rank}\{\varrho_{\kappa_0}\}=d-1$.

With such starting point, the squared speed   can be increased by increasing the purity to  $\kappa>\kappa_0$ via two strategies: (i)  Increasing some off-diagonal entries $\varrho_{ij}$ except $\varrho_{1d}$, but keeping the diagonal entries  fixed. (ii) Increasing the off-diagonal entry $\varrho_{1d}$ which, simultaneously, requires increasing one or both diagonal entries $\varrho_{11}$ and $\varrho_{dd}$ in order to maintain the positivity of $|M(1,d)|$. In this case, the diagonal elements $\varrho_{jj}$ for $j=2,\cdots,d-1$ should be appropriately   changed  to maintain the state normalized.
The tradeoff  between $\sum_{i=1}^{d} \varrho_{ii}^{2} $ and $|\varrho_{1d}|^2$  should be optimized in the sense that  no increase in purity should occur unless the speed increased  as much as possible. In the light of this, the restriction  $|\varrho_{1d}|^{2}\le \varrho_{11} \varrho_{dd}$ provides a better limit if we  choose  $|\varrho_{1d}|^{2}=\varrho_{11} \varrho_{dd}$ and  $\varrho_{11}=\varrho_{dd}$.
It turns out therefore  that the off-diagonal element $|\varrho_{1d}|$ can touch a new record as $\varrho_{11}=\varrho_{dd}=1/d+x$ with $x\ge 0$. Accordingly, the optimal state for $\kappa\ge \kappa_0$    should be  searched  among the following states
\begin{eqnarray}\label{Rho-Op-k>k0}\nonumber
\varrho_{\kappa} &=& 2\left(\frac{1}{d}+x\right)\ket{\Psi_1}\bra{\Psi_1} +\left(\frac{1}{d}-\frac{2x}{d-2}\right)\sum_{i=2}^{d-1}\ket{E_i}\bra{E_i} \\ \label{Rho-Op-k>k0}
&+& \left(\sum_{i=2}^{d-2}\sum_{j=i+1}^{d-1}\varrho_{ij}\ket{E_i}\bra{E_j}+\textrm{C.C.}\right),
\end{eqnarray}
which gives the first strategy for $ x=0 $ and the second strategy for $x \neq 0  $.
In writing equation above,  we have used Lemmas \ref{Lemma-Rho1d=Rhojd} and  \ref{Lemma-Rho1j=zero} to set   $\varrho_{1j}=\varrho_{jd}^\ast\e^{i\theta_j}=0$  for $j=2,\cdots,d-1$.

Next, from Eq. \eqref{SqSpeed-Unitary-1}, one can write the squared speed of the state \eqref{Rho-Op-k>k0} as
\begin{eqnarray}
v^2(\varrho_{\kappa})=2\left[(1/d+x)^2\omega_{1d}^2+\sum_{j=i+1}^{d-1}\sum_{i=2}^{d-2}|\varrho_{ij}|^2\omega^2_{ij}\right],
\end{eqnarray}
which  should be  maximized subject to the purity condition $\kappa-\Tr\varrho_{\kappa}^2=0$ where
\begin{eqnarray}\nonumber
\Tr\varrho_{\kappa}^2 &=&4\left(\frac{1}{d}+x\right)^2+(d-2)\left(\frac{1}{d}-\frac{2x}{d-2}\right)^2\\
&+& 2\sum_{j=i+1}^{d-1}\sum_{i=2}^{d-2}|\varrho_{ij}|^2.
\end{eqnarray}
Denoting by $\mu$ the corresponding Lagrange multiplier, and writing ${\tilde{v}}^2(\varrho_{\kappa})=v^2(\varrho_{\kappa})+\mu (\kappa-\Tr{\varrho_{\kappa}^2})$,    we have the following conditions  from $\partial {\tilde{v}}^2(\varrho_{\kappa})/\partial x=0$ and $\partial {\tilde{v}}^2(\varrho_{\kappa})/\partial |\varrho_{ij}|=0$,  respectively
\begin{itemize}
\item[C1:] \quad $x\left[\frac{2(d-1)}{d-2}\mu-\omega_{1d}^2\right]-\frac{1}{d}[\omega_{1d}^2-\mu]=0$,
\item[C2:] \quad $|\varrho_{ij}|\;[\mu-\omega^2_{ij}]=0$, \quad for \quad $2\le i< j \le d-1$.
\end{itemize}
From C1,  the positivity of $x$ requires  $\frac{d-2}{2(d-1)}\omega_{1d}^2\le \mu\le \omega_{1d}^2$. On the other hand, the condition $\varrho_{11}\le 1/2$ (or equivalently $x\le \frac{d-2}{2d}$),   gives $\mu \ge \omega_{1d}^2/2$, so that
\begin{equation}\label{mu-Bounds-4}
 \omega_{1d}^2/2\le \mu \le  \omega_{1d}^2.
\end{equation}
Both  inequalities are tight in the sense that the lower and upper bounds  saturate  for $\kappa=1$ and $\kappa=\kappa_0$, respectively. In the light of  this,  condition C1 reads
\begin{eqnarray}\label{x}
x=\frac{\omega_{1d}^2-\mu}{d\left[\frac{2(d-1)}{d-2}\mu-\omega_{1d}^2\right]}.
\end{eqnarray}
Condition C2, on the other hand,  gives $\varrho_{ij}=0$ or/and  $ \mu =\omega^2_{ij}$ for $2\le i< j \le d-1$. The Lagrange multiplier $\mu$ can be determined   only once,  implying that $ \mu =\omega^2_{ij}$ can be an acceptable solution for at most one pair of $i,j$\footnote{It may happens in some cases that the Bohr frequencies  $\omega_{ij}$ be equal for two or more pairs of $i,j$s, however, it follows from  $\sum_{i=2}^{d-2}\omega_{i,i+1}^2\le\omega_{2,d-1}^2$ that even in these cases the contribution of   $\omega_{2,d-1}$   will be more than the sum of the  others.} for which  $\omega_{1d}^2/2\le \omega_{ij}^2 \le  \omega_{1d}^2$, as such   most $\varrho_{ij}$s  should be vanished.  As $\omega_{ij}$ always satisfies the upper bound  $\omega_{ij}\le \omega_{1d}$, we need to check whether $\omega^2_{1d}/2\le \omega^2_{ij}$ or not.  However, as soon as the lower bound  is satisfied  by  one $\omega_{ij}$, it is also satisfied by  the largest one, namely $\omega_{2,d-1}$, which plays more important role  in the squared speed.  In view of this, we define $\gamma_{1}=\omega_{1d}^2/\omega_{2,d-1}^2>1$ which depends solely on the energy spectra  of the Hamiltonian, and  continue  our analysis in two different regimes; (i) $\gamma_{1}\ge 2$ and  (ii) $\gamma_{1}<2$.

\subsection*{Step 2 ($\gamma_1\ge 2,\quad \kappa\in[\kappa_0,1]$)}
If $\gamma_{1}\ge 2$, then $ \mu =\omega^2_{2,d-1}$ is  not an acceptable solution, so   $\varrho_{2,d-1} =0$. The same is true for all other  $\mu=\omega^2_{ij}$, accordingly,  $\varrho_{ij}=0$  for all possible values  $2\le i< j\le d-1$, and we arrive at a state given by  Eq. \eqref{rhoOP-1-2} for  the optimal state when  $\kappa\in[\kappa_0,1]$, i.e.,
\begin{eqnarray}\nonumber \label{rhoOP-1-3}
\varrho_{\kappa\in[\kappa_0,1]}&=&  2\left(\frac{1}{d}+x\right)\ket{\Psi_1}\bra{\Psi_1}\quad (\textrm{for } \gamma_1\ge 2) \\ \nonumber
&+&\left(\frac{1}{d}-\frac{2x}{d-2}\right)\sum_{i=2}^{d-1}\ket{E_i}\bra{E_i},
\end{eqnarray}
where $x$, given by Eq. \eqref{x-1},  is  the positive solution of the purity relation $\kappa=4(1/d+x)^2+(d-2)(1/d-2x/(d-2))^2$.
Using this solution of $x$ in Eq. \eqref{x}, the Lagrange multiplier $\mu$ can be determined  which satisfies  bounds \eqref{mu-Bounds-4} whenever  $\kappa\in[\kappa_0,1]$.

In this regime of $\gamma_1\ge 2$, the rank of the optimal state remains $d-1$ for all purity $\kappa\in[\kappa_0,1)$, but it collapses suddenly to $1$  at $\kappa=1$ where  the optimal state is described by the pure state $P_{\Psi_1}=\ket{\Psi_{1}}\bra{\Psi_1}$. This completes the proof  for  the optimal state for all range of purities in the regime $\gamma_{1}\ge 2$. When, on the other hand,  $\gamma_{1}<2$ we need to continue our search for the optimal state according to the following steps.

\subsection{Proof of   Theorem \ref{Theorem-Opt-Main-2}}\label{Appendix-Theorem-Opt-Main-2}
\subsection*{Step 3 ($\gamma_1<  2,\quad \kappa\in[\kappa_0,\kappa_1], \quad  \kappa\in[\kappa_1,\kappa_2]$)}
If $\gamma_{1} < 2$,  the solution $\mu=\omega_{2,d-1}^2$ does not violate Eq. \eqref{mu-Bounds-4}, implying  that   $\varrho_{2,d-1}$  can be nonzero subject to the positivity condition $|\varrho_{2,d-1}|^2\le \varrho_{22}\varrho_{d-1,d-1}=(1/d-2x/(d-2))^2$.      With $\mu=\omega^2_{2,d-1}=\omega_{1d}^2/\gamma_{1}$, condition C1 leads  to  $x=x_0$ where $x_0$ is given in Eq. \eqref{x0}
and satisfies   $0\le  x_0 \le (d-2)/(2d)$ as $1 \le \gamma_{1} \le  2$. Putting $x=x_0$  into
\begin{eqnarray}\nonumber
\kappa=4\left(\frac{1}{d}+x\right)^2+(d-2)\left(\frac{1}{d}-\frac{2x}{d-2}\right)^2+ 2|\varrho_{2,d-1}|^2,
\end{eqnarray}
and solving for $|\varrho_{2,d-1}|^2$ we get $|\varrho_{2,d-1}|^2=(\kappa-\kappa_1)/2$ where
\begin{eqnarray}\label{kappa1-2}
\kappa_1=4\left[1/d+x_0\right]^2+(d-2)\left[1/d-2x_0/(d-2)\right]^2.
\end{eqnarray}
In turn,  for  $\kappa \le \kappa_1$ the off-diagonal entry  $|\varrho_{2,d-1}|$ attains  its minimum  value $0$, accordingly, for $\kappa\in [\kappa_0,\kappa_1]$ the optimal  state is still served by $\varrho_{\kappa}$ of Eq. \eqref{rhoOP-1-2}  for which $\varrho_{2,d-1}=0$, so that
\begin{eqnarray}\nonumber
\varrho_{\kappa\in[\kappa_0,\kappa_1]}&=&  2\left(\frac{1}{d}+x\right)\ket{\Psi_1}\bra{\Psi_1}\quad (\textrm{for } \gamma_1 < 2) \\ \label{rhoOP-2-5}
&+&\left(\frac{1}{d}-\frac{2x}{d-2}\right)\sum_{i=2}^{d-1}\ket{E_i}\bra{E_i},
\end{eqnarray}
where $x$ is given by Eq. \eqref{x-1}.

On the other hand, the   condition  $|\varrho_{2,d-1}|^2\le \left[1/d-2x_0/(d-2)\right]^2$  provides us with    $\kappa_1 \le \kappa \le \kappa_2$, where
\begin{eqnarray}\label{kappa2-2}
\kappa_2 &=& \kappa_1+2\left[1/d-2x_0/(d-2)\right]^2.
\end{eqnarray}
The corresponding state  for $\kappa\in[\kappa_1,\kappa_2]$ then  turns to be  Eq. \eqref{rhoOP-2-1}.
It turns out also that  $\textrm{rank}\{\varrho_{\kappa}\}=d-1$ for $\kappa\in[\kappa_1,\kappa_2)$, but it is diminished by 1 for $\kappa=\kappa_2$ such that  $\textrm{rank}\{\varrho_{\kappa_2}\}=d-2$.

Note that, by  using $x_0$ from Eq. \eqref{x0},  Eqs. \eqref{kappa1} and \eqref{kappa2} for $\kappa_1$ and $\kappa_2$ can be obtained from Eqs. \eqref{kappa1-2} and \eqref{kappa2-2}, respectively.
Moreover,  $\kappa_1$ and $\kappa_2$ satisfy $\kappa_0\le \kappa_1 \le 1$ and $\kappa_0+2/d^2\le \kappa_2 \le 1$ (where $\kappa_0$ is given by Eq. \eqref{kappa0}) proportional  to  $0\le  x_0\le (d-2)/(2d)$ (or equivalently $1\le  \gamma_{1} \le  2$).

\subsection*{Step 4 ($\gamma_1< 2,\quad  \kappa\ge \kappa_2$)}
 We should continue this procedure until the purity tends to $\kappa=1$, corresponding to the optimal pure state $P_{\Psi_1}=\ket{\Psi_{1}}\bra{\Psi_1}$.
To obtain the optimal  state for $\kappa\ge \kappa_2$, we start with the optimal state \eqref{rhoOP-2-1} when  $\kappa=\kappa_2$, i.e.,
\begin{eqnarray}\nonumber
\varrho_{\kappa_2} &=& 2\left(\frac{1}{d}+x_0\right)\ket{\Psi_1}\bra{\Psi_1}+2\left(\frac{1}{d}-\frac{2x_0}{d-2}\right)\ket{\Psi_2}\bra{\Psi_2} \\
&+&\left(\frac{1}{d}-\frac{2x_0}{d-2}\right)\sum_{i=3}^{d-2}\ket{E_i}\bra{E_i} \quad (\textrm{for } \gamma_1 < 2).
\end{eqnarray}
Above,  $\ket{\Psi_{1}}=\frac{1}{\sqrt{2}}(\ket{E_1}+\e^{-i\theta_1}\ket{E_d})$ and $\ket{\Psi_{2}}=\frac{1}{\sqrt{2}}(\ket{E_2}+\e^{-i\theta_2}\ket{E_{d-1}})$ are two orthogonal pure states,  having maximum squared speed in the subspaces spanned by $\{\ket{E_1},\ket{E_d}\}$ and $\{\ket{E_2},\ket{E_{d-1}}\}$, respectively.

Starting with $\varrho_{\kappa_2}$, we  follow the method provided in Step 1, i.e.,  (i)  increasing some off-diagonal entries $\varrho_{ij}$ except $\varrho_{1d}$ and $\varrho_{2,d-1}$, but keeping the diagonal entries  fixed. (ii) Changing  the off-diagonal entries  $\varrho_{1d}$ and $\varrho_{2,d-1}$, by changing  the products  $\varrho_{11}\varrho_{dd}$ and $\varrho_{22}\varrho_{d-1,d-1}$, respectively,  in such a way that $|M(1,d)|$ and $|M(2,d-1)|$ remain nonnegative.
It turns out that  $\ket{\Psi_1^{\perp}}$ no longer belongs to $\textrm{Ker}\{\varrho_{\kappa\ge \kappa_2}\}$  unless  $|\varrho_{1d}|^{2}=\varrho_{11} \varrho_{dd}$ with   $\varrho_{11}=\varrho_{dd}$. Similarly, $\ket{\Psi_2^{\perp}}$ no longer belongs to $\textrm{Ker}\{\varrho_{\kappa\ge \kappa_2}\}$ unless  $|\varrho_{2,d-1}|^{2}=\varrho_{22} \varrho_{d-1,d-1}$ with  $\varrho_{22}=\varrho_{d-1,d-1}$.
Accordingly,  two  off-diagonal elements $|\varrho_{1d}|$ and $|\varrho_{2,d-1}|$ touch a new record as $\varrho_{11}=\varrho_{dd}=1/d+x_0 +y_1$  and $\varrho_{22}=\varrho_{d-1,d-1}=1/d-2x_0/(d-2)-y_2$ with $y_1,y_2\ge 0$. The optimal state for $\kappa\ge \kappa_2$    should be  searched  among the following states
\begin{eqnarray}\nonumber
\varrho_{\kappa\ge \kappa_2} &=& 2\left(\frac{1}{d}+x_0+y_1\right)\ket{\Psi_1}\bra{\Psi_1}  \quad (\textrm{for } \gamma_1 < 2)\\ \nonumber
&+&2\left(\frac{1}{d}-\frac{2x_0}{d-2}-y_2\right)\ket{\Psi_2}\bra{\Psi_2} \\ \nonumber
&+&\left(\frac{1}{d}-\frac{2x_0}{d-2}-\frac{2(y_1-y_2)}{d-4}\right)\sum_{i=3}^{d-2}\ket{E_i}\bra{E_i} \\  \label{Rho-Op-k>k2}
&+& \left(\sum_{j=i+1}^{d-2}\sum_{i=3}^{d-3}\varrho_{ij}\ket{E_i}\bra{E_j}+\textrm{C.C.}\right),
\end{eqnarray}
Further calculation can be done by optimizing  the squared speed $v^2(\varrho_{\kappa})$ subject to  the purity condition $\kappa-\Tr\varrho_{\kappa}^2=0$.
The calculations are, however,  more  tedious and long and we do not get involved in for a general $d$.

\begin{remark}
By specializing to the particular case of   $d\le 4$, the steps given above are exhaustive so one can obtain the optimal states \eqref{Optimum-d=3} and \eqref{Optimum-d=4} for $d=3$ and $d=4$, respectively. In particular, for $d=4$, we have $y_1=y_2$ where $y_1=\sqrt{(2\kappa-1)}/4-x_0$ is the positive solution of the purity relation $\kappa=\Tr\varrho_{\kappa}^2$. Equation  \eqref{Rho-Op-k>k2} leads therefore  to
\begin{eqnarray}\nonumber
\varrho_{\kappa\in [\kappa_2,1]} &=&  \frac{1}{2}\left(1+\sqrt{2\kappa-1}\right)\ket{\Psi_1}\bra{\Psi_1} \quad (\textrm{for } \gamma_1 < 2) \\  \label{Rho-Op-k>k2-d=4}
&+&\frac{1}{2}\left(1-\sqrt{2\kappa-1}\right)\ket{\Psi_2}\bra{\Psi_2},
\end{eqnarray}
which completes the optimal state for entire  range of purities $\kappa\in[1/4,1]$ when  $d=4$.
\end{remark}



\begin{thebibliography}{99}

\bibitem{BekensteinPRL1981} J. D. Bekenstein,  Energy cost of information transfer,  Phys. Rev. Lett. {\bf 46}, 623 (1981).

\bibitem{LloydN2000} S. Lloyd, Ultimate physical limits to computation,  Nature {\bf 406}, 1047 (2000).

\bibitem{GiovannettiNP2011} V. Giovannetti, S. Lloyd, and L. Maccone,   Advances in quantum metrology,  Nat. Photon. {\bf 5}, 222–9 (2011).

\bibitem{CanevaPRL2009} T. Caneva, M. Murphy, T. Calarco, R. Fazio,  S. Montangero, V. Giovannetti,  and G. E. Santoro,
Optimal control at the quantum speed limit,  Phys. Rev. Lett. {\bf 103},  240501 (2009).

\bibitem{MandelstamJP1945} L. Mandelstam  and  I. Tamm, The uncertainty relation between energy and time in nonrelativistic
quantum mechanics, J. Phys.,  {\bf 9}, 249  (1945).

\bibitem{MargolusPD1998} N. Margolus and L. B. Levitin,  The maximum speed of dynamical evolution,  Physica D {\bf 120}, 188 (1998).

\bibitem{LevitinPRL2009} L. B. Levitin and Y. Toffoli, Fundamental limit on the rate of quantum dynamics: The unified
bound is tight,  Phys. Rev. Lett. {\bf 103},  160502 (2009).

\bibitem{GiovannettiPRA2003} V. Giovannetti, S. Lloyd, and L. Maccone,  Quantum limits to dynamical evolution,  Phys. Rev. A {\bf 67}, 052109 (2003).

\bibitem{WuPRA2018} S. X. Wu  and C. S. Yu,  Quantum speed limit for a mixed initial state,  Phys. Rev. A {\bf 98},  042132 (2018).

\bibitem{CampaioliPRL2018} F. Campaioli, F. A.  Pollock, F. C.  Binder, and K. Modi,  Tightening quantum speed limits for
almost all states,  Phys. Rev. Lett. {\bf 120},  060409 (2018).

\bibitem{TaddeiPRL2013} M. M. Taddei, B. M. Escher, L. Davidovich, and R. L. de Matos Filho,  Quantum speed limit for
physical processes,  Phys. Rev. Lett. {\bf 110}, 050402 (2013).

\bibitem{delCampoPRL2013} A. del Campo, I. L. Egusquiza, M. B. Plenio,  and S. F. Huelga,  Quantum speed limits in open
system dynamics,  Phys. Rev. Lett. {\bf 110}, 050403 (2013).

\bibitem{DeffnerPRL2013} S. Deffner  and E. Lutz,  Quantum speed limit for non-Markovian dynamics,  Phys. Rev. Lett.
{\bf 111}, 010402 (2013).



\bibitem{GiovannettiEU2003} V. Giovannetti, S. Lloyd, and L. Maccone,  The role of entanglement in dynamical evolution,
Europhys. Lett. {\bf 62},  615 (2003).

\bibitem{BatlePRA2005} J. Batle, M. Casas, A. Plastino, and A. R. Plastino,  Connection between entanglement and the speed of quantum evolution,  Phys. Rev. A   {\bf 72},  032337 (2005).

\bibitem{ZanderJPA2007} C. Zander, A. R. Plastino, A. Plastino,  and M. Casas  Entanglement and the speed of evolution of
multi-partite quantum systems,  J. Phys. A: Math. Theor.  {\bf 40},  2861 (2007).

\bibitem{FrowisPRA2012} F. Fr\"{o}wis,   Kind of entanglement that speeds up quantum evolution,  Phys. Rev. A  {\bf 85}, 052127 (2012).

\bibitem{BorrasPRA2006} A. Borr\'{a}s, M. Casas, A. R. Plastino, and A. Plastino,   Entanglement and the lower bounds on the speed of quantum evolution,  Phys. Rev. A  {\bf 74}, 022326 (2006).

\bibitem{KupfermanPRA2008} J.  Kupferman and B. Reznik,   Entanglement and the speed of evolution in mixed states,  Phys. Rev. A  {\bf 78}, 042305 (2008).



\bibitem{CianciarusoPRA2017} M. Cianciaruso, S. Maniscalco, and G. Adesso,  Role of non-Markovianity and backflow of
information in the speed of quantum evolution,  Phys. Rev. A {\bf 96},  012105 (2017).

\bibitem{WuJPA2015} S. X. Wu, Y. Zhang, C. S. Yu, and H. S. Song,   The initial-state dependence of the quantum
speed limit, J. Phys. A  {\bf 48}, 045301 (2015).

\bibitem{LiuPRA2015} C. Liu, Z-Y Xu,  and S. Zhu,  Quantum-speed-limit time for multiqubit open systems,  Phys. Rev.
A {\bf 91}, 022102 (2015).





\bibitem{CampaioliNJP2022} F. Campaioli, C-s Yu, F. A. Pollock, and K.  Modi, Resource speed limits: maximal rate of resource variation,  New J. Phys. {\bf 24}, 065001 (2022).

\bibitem{MohanNJP2022} B. Mohan, S. Das, and A. K. Pati, Quantum speed limits for information and coherence,  New J. Phys. {\bf 24}, 065003 (2022).

\bibitem{PandeyPRA2023} V.  Pandey,  D. Shrimali,  B.  Mohan,  S.  Das.  and A.  K.  Pati, Speed limits on correlations in bipartite quantum systems Vivek,  Phys. Rev. A {\bf 107}, 052419 (2023).

\bibitem{RudnickiPRA2021} {\L}ukasz Rudnicki,  Quantum speed limit and geometric measure of entanglement,  Phys. Rev. A {\bf 104}, 032417 (2021).

\bibitem{FreyQINP2016} M. R. Frey, Quantum speed limits---primer, perspectives, and potential future directions,  Quantum Inf. Process.Phys.  {\bf 15}, 3919-3950 (2016).


\bibitem{CarliniPRL2006}  A.  Carlini, A.  Hosoya, T.  Koike,  and Y.  Okudaira,  Time-optimal quantum evolution, Phys. Rev. Lett. {\bf 96}, 060503 (2006).

\bibitem{CarliniJPA2008}  A.  Carlini, A.  Hosoya, T.  Koike,  and Y.  Okudaira,  Time optimal quantum evolution of mixed states, J. Phys. A: Math. Theor. {\bf 41}, 045303 (2008).

\bibitem{WangPRL2015} X Wang, M.  Allegra, K. Jacobs,  S.  Lloyd,  C.  Lupo, and M.  Mohseni,  Quantum brachistochrone curves as geodesics: obtaining accurate minimum-time protocols  for the control of quantum systems, Phys. Rev. Lett. {\bf 114}, 170501 (2015).

\bibitem{GengPRL2016}  J.  Geng,  Y.  Wu,  X.  Wang, K.  Xu,  F.  Shi,  Y.  Xie, X.  Rong,  and J. Du,  Experimental time-optimal universal control of spin qubits in solids, Phys. Rev. Lett. {\bf 117}, 170501 (2016).

\bibitem{LamPRX2021}  M. R. Lam,  N.  Peter,  T.  Groh,  W.  Alt,  C.  Robens,  D.  Meschede, A.  Negretti,  S.  Montangero,  T.  Calarco, and A.  Alberti,  Demonstration of quantum brachistochrones between distant states of an atom, Phys. Rev. X {\bf 11}, 011035 (2021).





\bibitem{DeffnerJPA2017} S. Deffner and S. Campbell,  Quantum speed limits: from Heisenberg’s uncertainty principle
to optimal quantum control, J. Phys. A: Math. Theor. {\bf 50}, 453001 (2017).

\bibitem{PiresPRX2016} D. P. Pires, M. Cianciaruso, L. C. Celeri, G. Adesso,  and D. O. Soares-Pinto,   Generalized geometric
quantum speed limits,  Phys. Rev. X {\bf 6},  021031 (2016).

\bibitem{CampaioliQ2019} F. Campaioli, F. A. Pollock, and K. Modi, Tight, robust, and feasible quantum speed limits for open dynamics, Quantum {\bf 3}, 168 (2019).

\bibitem{MarvianPRA2016}I. Marvian, R. W. Spekkens, and P. Zanardi, Quantum speed limits, coherence, and asymmetry, Phys. Rev. A {\bf 93}, 052331 (2016).

\bibitem{MondalPLA2016} D. Mondal, C. Datta,  and S. Sazim,  Quantum coherence sets the quantum speed limit for mixed
states,  Phys. Lett. A  {\bf 380},  689 (2016).

\bibitem{ShaoPRR2020} Y. Shao, B. Liu, M. Zhang, H. Yuan,  and J. Liu, Operational definition of a quantum speed limit, Phys. Rev. Research {\bf 2}, 023299 (2020)

\bibitem{BrodyPRR2019} D. C. Brody and B. Longstaff, Evolution speed of open quantum dynamics, Phys. Rev. Research  {\bf 2}, 033127 (2019).

\bibitem{GessnerPRA2018} M. Gessner and A. Smerzi, Statistical speed of quantum states: Generalized quantum Fisher information and Schatten speed, Phys. Rev. A {\bf 97}, 022109 (2018).

\bibitem{TextorIJTP1978} W. Textor,  A theorem on maximum variance,  Int. J. Theor. Phys. {\bf 17},  599-609 (1978).

\bibitem{Sakmann2011PRA} K. Sakmann, A. Streltsov, O. Alon O, L. Cederbaum,   Number fluctuations of cold spatially split bosonic objects, Phys. Rev. A {\bf 84}, 053622 (2011).

\bibitem{KarolPRA1998} K. \.{Z}yczkowski, P. Horodecki, A. Sanpera, and M. Lewenstein,Volume of the set of separable states, Phys. Rev. A  {\bf 58}, 883-892 (1998).

\bibitem{WoottersPRL1998} W. K. Wootters, Entanglement of Formation of an Arbitrary State of Two Qubits, Phys. Rev. Lett. {\bf 80}, 2245-2248  (1998).


\bibitem{PeresPRL1996} A. Peres, Separability Criterion for Density Matrices, Phys. Rev. Lett. {\bf 77}, 1413 (1996).

\bibitem{HorodeckiPLA1996} M.  Horodecki, P. Horodecki, and R. Horodecki, Separability of mixed states: necessary and sufficient conditions, Phys. Letts  A  {\bf 223}, 1-8 (1996).


\bibitem{PlenioPRL2014} T. Baumgratz, M. Cramer, and M. B. Plenio, Quantifying coherence,  Phys. Rev. Lett. {\bf 113}, 140401  (2014).


\end{thebibliography}
\end{document}